\documentclass[paper]{ieice}
\usepackage{amsthm}
\usepackage{amssymb}
\usepackage{amsfonts}
\usepackage{bm}

\usepackage[fleqn]{amsmath}
\usepackage{newtxtext}
\usepackage[varg]{newtxmath}
\usepackage{mathtools}
\usepackage{calc}
\DeclareFontFamily{U}{mathx}{}
\DeclareFontShape{U}{mathx}{m}{n}{<-> mathx10}{}
\DeclareSymbolFont{mathx}{U}{mathx}{m}{n}
\DeclareMathAccent{\widehat}{0}{mathx}{"70}
\DeclareMathAccent{\widecheck}{0}{mathx}{"71}

\mathtoolsset{showonlyrefs=true}

\usepackage[subrefformat=parens]{subcaption}
\usepackage{algorithmic,algorithm}
\usepackage{enumitem}

\field{A}
\title{
An LiGME Regularizer of Designated Isolated Minimizers \\
- An Application to Discrete-Valued Signal Estimation
}
\authorlist{%
 \authorentry[shoji@sp.ict.e.titech.ac.jp]{Satoshi Shoji}{m}{1}
 \MembershipNumber{}
 \authorentry[yata@sp.ict.e.titech.ac.jp]{Wataru Yata}{s}{1}
 \MembershipNumber{}
 \authorentry[kume@sp.ict.e.titech.ac.jp]{Keita Kume}{m}{1}
 \MembershipNumber{1920341}
 \authorentry[isao@sp.ict.e.titech.ac.jp]{Isao Yamada}{f}{1}
 \MembershipNumber{8510504}
}
\affiliate[affiliate label]{
  The authors are with the Department of Information and Communications Engineering, Institute of Science Tokyo, Tokyo, 152-8550 Japan
}
\theoremstyle{plain}
\newtheorem{theorem}{Theorem}
\newtheorem{lemma}{Lemma}
\newtheorem{proposition}{Proposition}
\newtheorem{fact}{Fact}
 
\theoremstyle{definition}
\newtheorem{problem}{Problem}
\newtheorem{scheme}{Scheme}
\newtheorem{remark}{Remark}
\newtheorem{claim}{Claim}
\newtheorem{definition}{Definition}
\newtheorem{example}{Example}

\theoremstyle{remark}

\def\d{\mathbf{d}}
\def\x{\mathbf{x}}
\def\xstar{\mathbf{x}^\star}
\def\y{\mathbf{y}}

\def\uu{\mathbf{u}}
\def\v{\mathbf{v}}
\def\w{\mathbf{w}}

\def\p{\mathbf{p}}
\def\dfrak{\bm{\mathfrak{d}}}
\def\A{\mathbf{A}}
\def\B{\mathbf{B}}
\def\I{\mathbf{I}}
\def\O{\mathbf{O}}

\def\M{\mathbf{M}}

\def\eps{\bm{\varepsilon}}
\def\ones{\mathbf{1}}
\def\SOAV{\text{SOAV}}
\def\LiGME{\text{LiGME}}
\def\cLiGME{\text{cLiGME}}
\def\qam{\text{qam}}
\def\psk{\text{psk}}

\newcommand{\myskip}[1]{
  \thickmuskip=#1\thickmuskip
  \medmuskip=#1\medmuskip
  \thinmuskip=#1\thinmuskip
  \arraycolsep=#1\arraycolsep
}

\newcommand{\abs}[1]{\lvert{#1}\rvert}
\newcommand{\babs}[1]{\bigl\lvert{#1}\bigr\rvert}
\newcommand{\norm}[1]{\lVert{#1}\rVert}
\newcommand{\bnorm}[1]{{\bigl\lVert{#1}\bigr\rVert}}
\newcommand{\opnorm}[1]{\lVert{#1}\rVert_{\mathrm{op}}}
\newcommand{\ev}[1]{\left\langle{#1}\right\rangle}

\newcommand{\argmin}{\operatornamewithlimits{argmin}}

\newcommand{\Prox}{\operatorname{Prox}}
\newcommand{\Id}{\mathrm{Id}}
\newcommand{\Fix}{\operatorname{Fix}}
\newcommand{\sgn}{\operatorname{sgn}}
\newcommand{\conv}{\operatorname{conv}}
\newcommand{\ri}{\operatorname{ri}}
\newcommand{\dom}{\operatorname{dom}}
\newcommand{\ran}{\operatorname{ran}}

\newcommand{\Span}{\operatorname{span}}
\newcommand{\aff}{\operatorname{aff}}
\newcommand{\Int}{\operatorname{int}}
\newcommand{\gra}{\operatorname{gra}}
\newcommand{\mybigtimes}{{\textstyle\bigtimes}}

\newcommand{\Hilbert}[2]{\left({#1},\ev{\cdot,\cdot}_{#2},\norm{\cdot}_{#2}\right)}

\begin{document}
\maketitle

\setcounter{footnote}{0}
\renewcommand{\thefootnote}{\arabic{footnote}}
\begin{summary}
  For a regularized least squares estimation of discrete-valued signals, we propose a Linearly involved Generalized Moreau Enhanced (LiGME) regularizer, as a nonconvex regularizer, of designated isolated minimizers.
  The proposed regularizer is designed as a Generalized Moreau Enhancement (GME) of the so-called sum-of-absolute-values (SOAV) convex regularizer.
  Every candidate vector in the discrete-valued set is aimed to be assigned to an isolated local minimizer of the proposed regularizer while the overall convexity of the regularized least squares model is maintained.
  Moreover, a global minimizer of the proposed model can be approximated iteratively by using a variant of the constrained LiGME (cLiGME) algorithm.
  To enhance the accuracy of the proposed estimation, we also propose a pair of simple modifications, called respectively an iterative reweighting and a generalized superiorization.
  Numerical experiments demonstrate the effectiveness of the proposed model and algorithms in a scenario of multiple-input multiple-output (MIMO) signal detection.
\end{summary}
\begin{keywords}
  LiGME regularizer,
  Designated isolated minimizer,
  Discrete-valued signal estimation,
  Overall convexity,
  MIMO signal detection
\end{keywords}

\section{Introduction}
\label{sec:introduction}

In this paper, we consider an inverse problem (see Problem~\ref{prob:discrete} below) for a discrete-valued signal estimation from a noisy linear observation of the target signal whose entries belong to a finite \emph{alphabet} $\mathfrak{A}$.
Such a problem arises widely in the field of signal processing, including generalized spatial modulation~\cite{liu_denoising_2013}, multiuser detection~\cite{sasahara_multiuser_2017,zhu_exploiting_2010}, and cognitive spectrum sensing~\cite{axell_spectrum_2012}, as well as discrete-valued image reconstruction~\cite{nikolova_estimation_1998,duarte_single-pixel_2008,bioglio_sparse_2014,tuysuzoglu_graph-cut_2015,sarangi_measurement_2022}.
\begin{problem}[A discrete-valued signal estimation problem]\label{prob:discrete}
\begin{equation}\label{eq:linear_regression}
  \text{Find~}
  \xstar\in\mathfrak{A}^{N}\subset\mathbb{R}^{N}
  \text{~such~that~}
  \y=\A\xstar+\eps,
\end{equation}
where
$\mathfrak{A}\coloneqq\left\{a^{\ev{1}},a^{\ev{2}},\ldots,a^{\ev{L}}\right\}\subset\mathbb{R}$ is a finite alphabet,
$\y\in\mathbb{R}^M$ is an observed vector,
$\A\in\mathbb{R}^{M\times N}$ is a known matrix,
and $\eps\in\mathbb{R}^M$ is noise
(Note: the complex version of this problem can also be reformulated as Problem~\ref{prob:discrete} essentially via simple $\mathbb{C}\rightleftarrows\mathbb{R}^{2}$ translation; see Section~\ref{sec:application}).
\end{problem}

The maximization of the likelihood function over $\mathfrak{A}^{N}$ is an ideal strategy for Problem~\ref{prob:discrete} if the statistical information of the target signal and noise is available.
However, such a naive maximization problem is NP-hard~\cite{verdu_computational_1989}, and its computational complexity increases exponentially as the signal dimension $N$ grows.
To circumvent this issue, in the last decades, regularized least squares models~\cite{zhu_exploiting_2010,wu_high-throughput_2016,kudeshia_total_2019,chen_manifold_2017,elgabli_two-stage_2017,hayakawa_convex_2017,iimori_robust_2021} have gained attention, where the regularizer is designed strategically based on a prior knowledge regarding the target $\xstar$.
We summarize 
the main idea of such approaches below:
\begin{scheme}[A scheme for~\eqref{eq:linear_regression} via regularized least squares]
  \label{scheme}
  \,
  \textbf{Step 1 (Solve a regularized least squares model):}
  \begin{equation}\label{eq:regularized_least_squares}
    \text{Find~}\x^{\lozenge}\in
    \argmin_{\x\in\mathcal{C}}J_{\Theta}(\x)\coloneqq\frac12\norm{\y-\A\x}_2^2+\mu\Theta(\x),
  \end{equation}
  where the constraint set $\mathcal{C}\subset\mathbb{R}^{N}$ is chosen usually as a connected superset (e.g., the convex hull ) of $\mathfrak{A}^{N}$, $\Theta:\mathbb{R}^N\to\mathbb{R}$ is a regularizer, and $\mu\in\mathbb{R}_{++}$ is a regularization parameter.

  \noindent\textbf{Step 2 (Find a nearest discrete-valued signal):}
  With the tentative estimate $\x^{\lozenge}\in\mathcal{C}\subset\mathbb{R}^{N}$, compute the final estimate $\x^{\blacklozenge}\in\mathfrak{A}^{N}$ of $\xstar$ in~\eqref{eq:linear_regression} as
  \begin{equation}\label{eq:projection}
    \x^{\blacklozenge}= P_{\mathfrak{A}^{N}}\bigl(\x^{\lozenge}\bigr)\in\argmin_{\dfrak\in\mathfrak{A}^{N}}\bnorm{\dfrak-\x^{\lozenge}}_2,
  \end{equation}
  where $P_{\mathfrak{A}^{N}}:\mathbb{R}^N\to\mathfrak{A}^{N}:\x\mapsto P_{\mathfrak{A}^{N}}(\x)$ is defined to choose one of nearest vectors in $\mathfrak{A}^{N}$ from $\x$.

\end{scheme}

For Problem~\ref{prob:discrete}, the regularizer $\Theta$ in~\eqref{eq:regularized_least_squares} is desired to enjoy a certain attractivity into the set $\mathfrak{A}^{N}$.
More precisely, $\Theta$ is expected to have isolated local minimizers (see Definition~\ref{dfn:isolated}) at points in $\mathfrak{A}^{N}$.
For design of such a contrastive regularizer $\Theta$, we must consider the computational tractability of finding a global minimizer of $J_{\Theta}$ in~\eqref{eq:regularized_least_squares} over $\mathcal{C}$ as well.

Toward such a regularizer, we first review existing regularizers for utilization in~\eqref{eq:regularized_least_squares}.
The so-called \emph{sum-of-absolute-values (SOAV)} regularizer~\cite{nagahara_discrete_2015,hayakawa_convex_2017}
{
  \myskip{0.3}
  \begin{align}
     & \hspace{-2em}({\x}\coloneqq({x}_1,{x}_2,\ldots,{x}_{N})^{\top}\in\mathbb{R}^{N}) \\
     & \label{eq:theta_SOAV}
    \hspace{-2em}\Theta_{\SOAV}(\x)
    \coloneqq\sum_{l=1}^{L} \bnorm{\x-a^{\ev{l}}\ones_{N}}_{\bm{\omega}^{\ev{l}};1}
    \coloneqq \sum_{l=1}^{L}\sum_{n=1}^{N} \omega^{\ev{l}}_{n}{\babs{x_{n}-a^{\ev{l}}}}
  \end{align}
}%
has been proposed with weighting vectors
$\bm{\omega}^{\ev{l}} \coloneqq \left(\omega^{\ev{l}}_{1},\omega^{\ev{l}}_{2},\ldots,\omega^{\ev{l}}_{N}\right)^{\top} \in \mathbb{R}_{++}^{N}\ (l=1,2,\ldots,L)$
satisfying
$\sum_{l=1}^{L}\omega^{\ev{l}}_{n}=1$ $(n=1,2,\ldots,N)$,
where, for $\uu\coloneqq(u_{1},u_{2},\ldots,u_{N})^{\top}\in\mathbb{R}^{N}$, 
$\norm{\uu}_{\bm{\omega}^{\ev{l}};1} \coloneqq \sum_{n=1}^{N} \omega^{\ev{l}}_{n} \abs{u_{n}}$
stands for its weighted $\ell_{1}$-norm associated with
$\bm{\omega}^{\ev{l}}$.
Although $J_{\Theta_{\SOAV}}$ is convex from the convexities of $\frac{1}{2}\norm{y-A\cdot}_{2}^{2}$ and $\Theta_{\SOAV}$ in~\eqref{eq:theta_SOAV},
Fig.~\ref{fig:theta_LiGME_simple}~\subref{fig:theta_SOAV} suggests that penalization by convex $\Theta_{\SOAV}$ is not contrastive enough to express the alphabet as its local minimizers because any point $\dfrak\in\mathfrak{A}^{N}$ is never positioned at isolated local minimizer of $\Theta_{\SOAV}$.
Indeed, convex regularizers cannot have multiple isolated local minimizers because any
convex combination of local minimizers is a global minimizer.
For the discrete-valued signal estimation, nonconvex regularizers that are more contrastive than $\Theta_{\SOAV}$ have been desired for use in~\eqref{eq:regularized_least_squares} of~Step~1 (see, e.g., Fig.~\ref{fig:theta_LiGME_simple}~\subref{fig:theta_LiGME_id} and Fig.~\ref{fig:theta_LiGME_practical}).
A nonconvex variant $\Theta_{\SOAV}^{\ev{p}}(\x) \coloneqq \sum_{l=1}^{L} \bnorm{\x - a^{\ev{l}}\ones_{N}}_{\bm{\omega}^{\ev{l}};p}^{p}$ $(0 < p < 1)$
of $\Theta_{\SOAV}$ has also been used in~\cite{hayakawa_discrete-valued_2019}, where $\norm{\uu}_{\bm{\omega}^{\ev{l}};p}\coloneqq\left(\sum_{n=1}^{N}\omega^{\ev{l}}_{n}\abs{u_{n}}^{p}\right)^{\frac{1}{p}}$
and $\Theta_{\SOAV}^{\ev{1}}$ reproduces $\Theta_{\SOAV}$.
However, any algorithm producing a convergent sequence to a global minimizer of $J_{\Theta_{\SOAV}^{\ev{p}}}$ has not yet been established mainly because of the severe nonconvexities of $\Theta_{\SOAV}^{\ev{p}}$ and $J_{\Theta_{\SOAV}^{\ev{p}}}$.
An exceptional example, which enjoys both contrastiveness of $\Theta$ and computational tractability of minimization of $J_{\Theta}$, is found in~\cite{nikolova_estimation_1998} for a special case of Problem~\ref{prob:discrete} with $L=2$ and $\A$ of full column rank.
\cite{nikolova_estimation_1998} can be seen as a pioneering work in the study of convexity-preserving nonconvex regularizers~\cite{blake_visual_1987,zhang_nearly_2010,selesnick_sparse_2017,abe_linearly_2020,yata_constrained_2022}.

\begin{figure}[t]
  \def\contourheightA{3.5cm}
  \centering
  \begin{minipage}[t]{0.49\linewidth}
    \hspace*{-2em}
    \centering
    \includegraphics[keepaspectratio, height=\contourheightA]{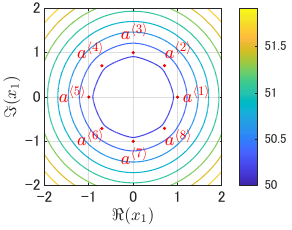}
    \subcaption{$\Theta_{\SOAV}$}
    \label{fig:theta_SOAV}
  \end{minipage}\hfill
  \begin{minipage}[t]{0.49\linewidth}
    \hspace*{-1em}
    \centering
    \includegraphics[keepaspectratio, height=\contourheightA]{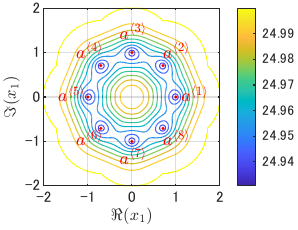}
    \subcaption{$\Theta_{\LiGME}$ with
      $\myskip{0.2}\B^{\ev{l}}=(1/\sqrt{L})\I_{50}$}
    \label{fig:theta_LiGME_id}
  \end{minipage}
  \caption{
    Illustrations of the values of $\Theta$ in the case where $\mathfrak{A}\coloneqq\left\{a^{\ev{l}}\coloneqq\exp[j(l-1)\pi/4]\mid l=1,2,\ldots,8\eqqcolon L\right\}\subset\mathbb{C}\equiv\mathbb{R}^{2}$ and $\omega^{\ev{l}}_{n}=1/8$ $(l=1,2,\ldots,8;\,n=1,2,\ldots 50\eqqcolon N)$.
    (a)~$\Theta_{\SOAV}(\x)$,
    (b)~a proposed nonconvex enhancement $\Theta_{\LiGME}$ of $\Theta_{\SOAV}$ by GME~\eqref{eq:GME} with $\B^{\ev{l}}=(1/\sqrt{L})\I_{50}$ $(l=1,2,\ldots,8)$.
    For visualization, we set $x_{n}=0$ $(n=2,3,\ldots,50)$.
  }
  \label{fig:theta_LiGME_simple}
\end{figure}

\begin{figure}[t]
  \def\contourheightB{3.45cm}
  \centering
  \begin{minipage}[t]{0.49\linewidth}
    \hspace*{-2em}
    \centering
    \includegraphics[keepaspectratio, height=\contourheightB]{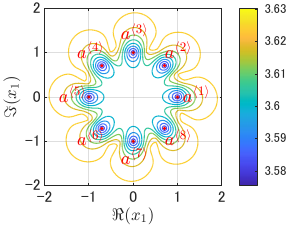}
    \subcaption{Overall view}
    \label{fig:theta_LiGME_practical_zentai}
  \end{minipage}
  \hfill
  \begin{minipage}[t]{0.49\linewidth}
    \hspace*{-1em}
    \centering
    \includegraphics[keepaspectratio, height=\contourheightB]{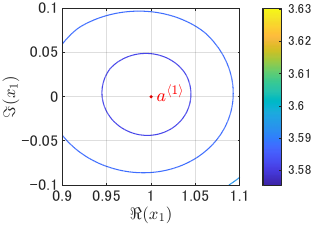}
    \subcaption{Enlarged view around $a_1$}
    \label{fig:theta_LiGME_practical_kakudai}
  \end{minipage}
  \caption{Illustrations of the values of $\Theta_{\LiGME}(\x)$ in~Eq.~\eqref{eq:theta_LiGME} with $\mu=0.01$ and $\bigl(\B^{\ev{l}}\bigr)_{l=1}^{L}$ in~\eqref{eq:Bl} under the setting of Section~\ref{sec:numerical_experiment}.
  $\mathfrak{A}\subset\mathbb{C}$ and $\bm{\omega}^{\ev{l}}\in\mathbb{R}_{++}^{50}$ $(l=1,2,\ldots,L)$ are the same as those in Fig.~\ref{fig:theta_LiGME_simple}.
  For visualization, we set $x_{n}=0$ $(n=2,3,\ldots,50)$.
  }
  \label{fig:theta_LiGME_practical}
\end{figure}

In this paper, along such a modern view~\cite{selesnick_sparse_2017,abe_linearly_2020,yata_constrained_2022}, we propose a contrastive nonconvex regularizer
\begin{equation} \label{eq:theta_LiGME}
  \Theta_{\LiGME}(\x)\coloneqq
  \sum_{l=1}^{L}\bigl(\norm{\cdot}_{\bm{\omega}^{\ev{l}};1}\bigr)_{\B^{\ev{l}}}\bigl(\x-a^{\ev{l}}\ones_{N}\bigr),
\end{equation}
where, for each $l=1,2,\ldots,L$, $\bigl(\norm{\cdot}_{\bm{\omega}^{\ev{l}};1}\bigr)_{\B^{\ev{l}}}$ is a nonconvex enhancement (known as the \emph{Generalized Moreau Enhancement} (GME)) of $\norm{\cdot}_{\bm{\omega}^{\ev{l}};1}$
with a tunable matrix
$\B^{\ev{l}}$ (see~\eqref{eq:GME} in Section~\ref{sec:cLiGME}). 
$\Theta_{\LiGME}$ reproduces $\Theta_{\SOAV}$ in~\eqref{eq:theta_SOAV} with $\B^{\ev{l}}=\O$ (zero matrix).
As will be shown in Theorem~\ref{thm:isolated} in Section~\ref{sec:proposed_regularizer}, $\Theta_{\LiGME}$ with a simple design of tuning matrices $\bigl(\B^{\ev{l}}\bigr)_{l=1}^{L}$ has designated isolated local minimizers at every vector in $\mathfrak{A}^{N}$, which tells us that $\Theta_{\LiGME}$ can work, in~\eqref{eq:regularized_least_squares}, as a contrastive regularizer desired for Problem~\ref{prob:discrete} (see Fig.~\ref{fig:theta_LiGME_simple} for visual comparison).
Moreover, as seen in Lemma~\ref{lem:LiGME_regularizer} in Section~\ref{sec:proposed_regularizer}, the proposed regularizer can be seen as a special instance of the \emph{LiGME regularizer}~\cite{abe_linearly_2020}.
Hence, by imposing the \emph{overall convexity condition} for the LiGME model~\cite[Prop. 1]{abe_linearly_2020}, $J_{\Theta_{\LiGME}}$ can be made convex with a strategic choice of $\bigl(\B^{\ev{l}}\bigr)_{l=1}^{L}$
(see~\eqref{eq:occ_special} and~\eqref{eq:choice_Bl} in Section~\ref{sec:proposed_algorithm}).
Notably, Fig.~\ref{fig:theta_LiGME_practical} suggests that each point in $\mathfrak{A}$ is located almost exactly at an isolated local minimizer of $\Theta_{\LiGME}$ even under the overall convexity condition of $J_{\Theta_{\LiGME}}$.

Regarding the tractability of the minimization of $J_{\Theta_{\LiGME}}$ under the overall convexity of $J_{\Theta_{\LiGME}}$, we propose a variant of cLiGME algorithm~\cite{kitahara_multi-contrast_2021,yata_constrained_2022,yata_imposing_2024} (see Algorithm~\ref{alg}) that produces a convergent sequence to a global minimizer of $J_{\Theta_{\LiGME}}$ over a simple closed convex set $\mathcal{C}$.
We note that the proposed convergence analysis (see Proposition~\ref{prop:cLiGME_algorithm} and Theorem~\ref{thm:proposed_algorithm}) is shown under a weaker condition than the condition imposed in~\cite{abe_linearly_2020,kitahara_multi-contrast_2021,yata_constrained_2022,yata_imposing_2024}.
Numerical experiments in a scenario of MIMO signal detection demonstrate that the proposed cLiGME model outperforms the SOAV model~\cite{hayakawa_convex_2017}, i.e., the model~\eqref{eq:regularized_least_squares} with $\Theta\coloneqq\Theta_{\SOAV}$ in~\eqref{eq:theta_SOAV}.

We also propose two heuristic modifications (see \ref{appendix:modification}) in the process of the proposed algorithm (Algorithm~\ref{alg}): \emph{iterative reweighting} (inspired by~\cite{hayakawa_reconstruction_2018}) and \emph{generalized superiorization} (inspired by~\cite{censor_perturbation_2010,fink_superiorized_2023}).
The iterative reweighting adaptively adjusts the weights $\omega^{\ev{l}}_{n}$, used in $\Theta_{\LiGME}$, according to the previous estimate, and the generalized superiorization introduces strategic perturbations to move the latest estimate closer to one of nearest vectors in $\mathfrak{A}^{N}$ at each iteration.
The effectiveness of these modifications is confirmed via numerical experiments.

The preliminary short version is presented%
\footnote{
  Compared to~\cite{shoji_discrete-valued_2024}, the following three discussions are newly included in this paper:
  (i) The proposed regularizer $\Theta_{\LiGME}$ has designated isolated local minimizers under certain conditions (see~Definition~\ref{dfn:isolated} and~Theorem~\ref{thm:isolated});
  (ii) $\Theta_{\LiGME}$ is a special instance of the LiGME regularizers~\cite{abe_linearly_2020} (see~Lemma~\ref{lem:LiGME_regularizer});
  (iii) A convergence analysis of the proposed variant of cLiGME algorithm is given in Proposition~\ref{prop:cLiGME_algorithm}.

}
in~\cite{shoji_discrete-valued_2024}.

\noindent\textbf{Notation}.\,
$\mathbb{N}, \mathbb{R}, \mathbb{R}_{+},\mathbb{R}_{++}$ and $\mathbb{C}$ denote respectively the set of all positive integers, all real numbers, all nonnegative real numbers, all positive real numbers and all complex numbers.
$j\in\mathbb{C}$ stands for the imaginary unit, and $\Re(\cdot)$ and  $\Im(\cdot)$ stand respectively for the real and imaginary parts of complex number, vector, and matrix.

Let%
\footnote{
  In the rest of this paper, we use basic terminology from functional analysis (for example, $\mathcal{B}(\mathcal{H},\mathcal{K})$ denotes the set of all linear operator from $\mathcal{H}$ to $\mathcal{K}$) to maintain a balance between generality and concreteness while ensuring clarity.
  Readers who are not familiar with these concepts may refer to, e.g.,~\cite{luenberger_optimization_1997,yamada_kougaku_2009}.
}
$\Hilbert{\mathcal{H}}{\mathcal{H}}$ and $\Hilbert{\mathcal{K}}{\mathcal{K}}$ be finite dimensional real Hilbert spaces.
The set of all bounded linear operators from $\mathcal{H}$ to $\mathcal{K}$ is denoted by $\mathcal{B}(\mathcal{H},\mathcal{K})$.
$\Id\in\mathcal{B}(\mathcal{H},\mathcal{H})$ and $O_{\mathcal{H}}\in\mathcal{B}(\mathcal{H},\mathcal{H})$ respectively stand for the identity operator and the zero operator.
For $\mathcal{L}\in\mathcal{B}(\mathcal{H},\mathcal{K})$, $\mathcal{L}^{\ast}\in\mathcal{B}(\mathcal{K},\mathcal{H})$ denotes the adjoint operator of $\mathcal{L}$ (i.e., $(\forall (x,y)\in\mathcal{H}\times\mathcal{K})$ $\ev{\mathcal{L}x,y}_{\mathcal{K}}=\ev{x,\mathcal{L}^{\ast} y}_{\mathcal{H}}$).
For a self-adjoint operator $\mathcal{L}\in\mathcal{B}(\mathcal{H},\mathcal{H})$ (i.e., $\mathcal{L}^{\ast}=\mathcal{L}$),
we express the positive definiteness and the positive semidefiniteness of $\mathcal{L}$ respectively by
$\mathcal{L}\succ O_{\mathcal{H}}$ and $\mathcal{L}\succeq O_{\mathcal{H}}$.
Any $\mathcal{L}\succ O_{\mathcal{H}}$ defines a new Hilbert space
$\Hilbert{\mathcal{H}}{\mathcal{L}}$
equipped with an inner product
$\ev{\cdot,\cdot}_{\mathcal{L}}:\mathcal{H}\times\mathcal{H}\to\mathbb{R}:(x,y)\mapsto\ev{x,\mathcal{L}y}_{\mathcal{H}}$ and its induced norm
$\norm{\cdot}_{\mathcal{L}}:\mathcal{H}\to\mathbb{R}:x\mapsto\sqrt{\ev{x,x}_{\mathcal{L}}}$.
The range of $\mathcal{L}\in\mathcal{B}(\mathcal{H},\mathcal{K})$ is $\ran(\mathcal{L})=\{\mathcal{L}x\in\mathcal{K}\mid x\in\mathcal{H}\}$.
The operator norm of $\mathcal{L}\in\mathcal{B}(\mathcal{H},\mathcal{K})$ is denoted by $\opnorm{\mathcal{L}}\coloneqq\sup_{x\in\mathcal{H},\norm{x}_{\mathcal{H}}\leq1}\norm{\mathcal{L}x}_{\mathcal{K}}$.

A set $S\subset\mathcal{H}$ is said to be convex if $(1-\theta)x+\theta y\in S$ for all $(x,y,\theta)\in S\times S\times [0,1]$.
For a set $S\subset\mathcal{H}$,
(i)~the linear span of $S$ is
$\Span(S)\coloneqq \bigl\{\sum_{i=1}^{n}\alpha_{i} x_{i}\in\mathcal{H}\mid n\in\mathbb{N}, \alpha_{i}\in\mathbb{R},x_{i}\in S\bigr\}$,
(ii)~the affine hull of $S$ is
$\aff(S)\coloneqq x+\Span(S-x)$, where $x\in S$, $S-x\coloneqq \{w-x\in\mathcal{H}\mid w\in S\}$, and $x+\Span(S-x)=y+\Span(S-y)$ holds true for all $x,y\in S$,
(iii)~the convex hull $\conv(S)$ of $S$ is the smallest convex set containing $S$, and
(iv)~the interior of $S$ is
$\Int(S)\coloneqq \{x\in S\mid (\exists\rho>0)\,B(x;\rho)\subset S\}$,
where $B(x;\rho)\coloneqq\{y\in\mathcal{H}\mid\norm{y-x}_{\mathcal{H}}<\rho\}$ is an open ball of radius $\rho>0$ and center $x\in\mathcal{H}$.
The relative interior of a convex set $S$ is $\ri(S)\coloneqq\{x\in S\mid (\exists\rho>0)\,B(x;\rho)\cap \aff(S)\subset S\}$.
Clearly, $\Int(S)\subset\ri(S)$ holds true for every convex set $S\subset\mathcal{H}$.
The power set of $\mathcal{H}$, denoted by $2^{\mathcal{H}}$, is the collection of all subsets of $\mathcal{H}$, i.e., $2^{\mathcal{H}}\coloneqq\{S\mid S\subset\mathcal{H}\}$.

Multiple Hilbert spaces $\Hilbert{\mathcal{H}^{\ev{l}}}{\mathcal{H}^{\ev{l}}}$ $(l=1,2,\ldots,L)$
can be used to build a new Hilbert space
$\bm{\mathcal{H}}\coloneqq\bigl(\mybigtimes_{l=1}^{L}\mathcal{H}^{\ev{l}}\bigr)\coloneqq
  \mathcal{H}^{\ev{1}}\times\mathcal{H}^{\ev{2}}\times\cdots\times\mathcal{H}^{\ev{L}}=
  \bigl\{\bigl(x^{\ev{l}}\bigr)_{l=1}^{L}\coloneqq \bigl(x^{\ev{1}},x^{\ev{2}},\ldots,x^{\ev{L}}\bigr)\,\bigl|\, x^{\ev{l}}\in\mathcal{H}^{\ev{l}}\,(l=1,2,\ldots,L)\bigr\}$
equipped with (i)~the addition
$\bm{\mathcal{H}}\times\bm{\mathcal{H}}\to\bm{\mathcal{H}}:
  \left(\bigl(x^{\ev{l}}\bigr)_{l=1}^{L},\bigl(y^{\ev{l}}\bigr)_{l=1}^{L}\right)
  \mapsto
  \bigl(x^{\ev{l}}+y^{\ev{l}}\bigr)_{l=1}^{L}$,
(ii)~the scalar multiplication
$\mathbb{R}\times\bm{\mathcal{H}}\to\bm{\mathcal{H}}:
  \left(\alpha,\bigl(x^{\ev{l}}\bigr)_{l=1}^{L}\right)
  \mapsto
  \bigl(\alpha x^{\ev{l}}\bigr)_{l=1}^{L}$ and
(iii)~the inner product
$\bm{\mathcal{H}}\times\bm{\mathcal{H}}\to\mathbb{R}:
  \left(\bigl(x^{\ev{l}}\bigr)_{l=1}^{L},\bigl(y^{\ev{l}}\bigr)_{l=1}^{L}\right)
  \mapsto
  \ev{\bigl(x^{\ev{l}}\bigr)_{l=1}^{L},\bigl(y^{\ev{l}}\bigr)_{l=1}^{L}}_{\bm{\mathcal{H}}}\coloneqq \sum_{l=1}^{L}\ev{x^{\ev{l}},y^{\ev{l}}}_{\mathcal{H}^{\ev{l}}}$ 
  and its induced norm
$\bm{\mathcal{H}}\to\mathbb{R}:\x\mapsto\norm{\x}_{\bm{\mathcal{H}}}\coloneqq\sqrt{\ev{\x,\x}_{\bm{\mathcal{H}}}}$.

For discussion in Euclidean space, we use boldface letters to express vectors and normal font letters to express scalars.
$\mathbb{R}^{n}$ should be understood as a Euclidean space with the standard inner product $\ev{\cdot,\cdot}$ and the Euclidean norm $\norm{\cdot}_{2}$.
For a matrix expression $\M\in\mathbb{R}^{m\times n}$ of a linear operator $M\in\mathcal{B}(\mathbb{R}^{n},\mathbb{R}^{m})$, $\M^{\top}\in\mathbb{R}^{n\times m}$ denotes the transpose of $\M$ and corresponds to the adjoint $M^{\ast}\in\mathcal{B}(\mathbb{R}^{m},\mathbb{R}^{n})$ of $M$.
The symbols $\O_{m\times n}\in\mathbb{R}^{m\times n}$, $\I_{n}\in\mathbb{R}^{n\times n}$ and $\ones_{n}\in\mathbb{R}^{n}$ respectively stand for the zero matrix, the identity matrix and the all-one vector.
$\sgn:\mathbb{R}\to\{-1,0,1\}$ stands for the sign function
(Note: $\sgn(0)$ is understood as $0$).

\section{Preliminary}
\label{sec:preliminary}

\subsection{Selected elements of convex analysis and characterization of local minimizers}
\label{sec:convex_analysis}

A function $f:\mathcal{H}\to (-\infty,\infty]$ is said to be
(i)~\emph{proper} if $\dom f\coloneqq\{x\in\mathcal{H}\mid f(x)<\infty\}\neq \emptyset$,
(ii)~\emph{lower semicontinuous}
if $ \{ x\in \mathcal{H} \mid f(x)\leq \alpha\}$
is closed for every $\alpha\in \mathbb{R}$, and
(iii)~\emph{convex} if
$f(\theta x + (1-\theta) y)\leq \theta f(x)+(1-\theta) f(y) $
for all $(x,y,\theta)\in\mathcal{H}\times\mathcal{H}\times[0,1]$.
The set of all proper lower semicontinuous convex functions is denoted by $\Gamma_{0}(\mathcal{H})$.
For example, the following \emph{indicator function} $\iota_{C}$ of a nonempty closed convex set $C\subset\mathcal{H}$ belongs to $\Gamma_{0}(\mathcal{H})$:
\begin{equation}
  \label{eq:indicator}
  \iota_{C}:\mathcal{H}\to(-\infty,\infty]:
  x\mapsto
  \begin{cases}
    0,       & \text{if }x\in C,    \\
    +\infty, & \text{if }x\notin C.
  \end{cases}
\end{equation}

$f\in\Gamma_{0}(\mathcal{H})$ is said to be \emph{coercive} if $f(x)\to\infty$ as $\norm{x}_{\mathcal{H}}\to\infty$.
It is well-known that~\cite[Prop. 11.15]{bauschke_convex_2017} $f\in\Gamma_{0}(\mathcal{H})$ has a (global) minimizer (see Definition~\ref{dfn:isolated} (a)) over a closed convex subset $C\subset\mathcal{H}$, i.e., $f+\iota_{C}$ has a minimizer in $\mathcal{H}$, if $f$ is coercive and $(C\cap \dom f)\neq\emptyset$.

\begin{definition}[Characterizations of minimizers{~\cite[Sec. 2.1]{nocedal_numerical_2006}}]
  \label{dfn:isolated}
  Let
  $f:\mathcal{H} \to (-\infty,\infty]$ be
  a proper function.
  \begin{enumerate}[label=(\alph*), itemsep=1ex, leftmargin=*]
    \item
          A point $x^{\heartsuit}\in \dom f\subset\mathcal{H}$ is said to be a \emph{(global) minimizer} of $f$ if
          $(\forall x\in \mathcal{H})\,f(x^{\heartsuit})\leq f(x)$.
    \item
          A point
          $x^{\heartsuit} \in f$
          is said to be a \emph{local minimizer} of
          $f$
          if there exists an open neighborhood
          $\mathcal{N}_{x^{\heartsuit}} \subset \mathcal{H}$
          of
          $x^{\heartsuit}$
          such that
          $(\forall x \in \mathcal{N}_{x^{\heartsuit}})\,f(x^{\heartsuit})\leq f(x)$.
    \item
          A point
          $x^{\heartsuit} \in f$
          is said to be an \emph{isolated local minimizer} of
          $f$
          if there exists an open neighborhood
          $\mathcal{N}_{x^{\heartsuit}} \subset \mathcal{H}$
          of
          $x^{\heartsuit}$
          such that
          $x^{\heartsuit}$
          is the only local minimizer of
          $f$
          in
          $\mathcal{N}_{x^{\heartsuit}}$.
  \end{enumerate}
\end{definition}

\begin{lemma}[Equivalent conditions of isolated local minimizers]\label{lemma:separable}
  Let
  $f:\mathcal{H}\coloneqq \mathbb{R}^{N} \to (-\infty,\infty]$ be
  a proper function.
  Assume that
  $f$
  is the sum of univariate proper functions
  $f_{n}:\mathbb{R} \to (-\infty,\infty]\ (n=1,2,\ldots,N)$,
  i.e.,
  \begin{equation}
    (\x=(x_{1},x_{2},\ldots,x_{N})^{\top} \in \mathbb{R}^{N})
    \quad
    f(\x) = \sum_{n=1}^{N} f_{n}(x_{n}).
  \end{equation}
  Then, the following hold:
  \begin{enumerate}[label=(\alph*)]
    \item \label{item:separable:local}
          $\x^{\heartsuit}\coloneqq(x^{\heartsuit}_{1},x^{\heartsuit}_{2},\ldots,x^{\heartsuit}_{n}) \in f$
          is a local minimizer of
          $f$
          if and only if every
          $x^{\heartsuit}_{n}\ (n=1,2,\ldots,N)$
          is a local minimizer of
          $f_{n}$.
    \item \label{item:separable:isolated}
          $\x^{\heartsuit}\coloneqq(x^{\heartsuit}_{1},x^{\heartsuit}_{2},\ldots,x^{\heartsuit}_{n}) \in f$
          is an isolated local minimizer of
          $f$
          if and only if every
          $x^{\heartsuit}_{n}\ (n=1,2,\ldots,N)$
          is an isolated local minimizer of
          $f_{n}$.
  \end{enumerate}
  \begin{proof}
    See~\ref{appendix:proof_separable}.
  \end{proof}
\end{lemma}

\begin{definition}[Tools of convex analysis]
  \label{dfn:tools}
  \,
  \begin{enumerate}[label=(\alph*), leftmargin=*, itemsep=1ex]
    \item (Legendre-Fenchel transform~\cite[Def. 13.1]{bauschke_convex_2017}).
          For $f\in\Gamma_{0}(\mathcal{H})$, the function definded by
          \begin{equation}\label{eq:conjugate}
            f^{\ast}:\mathcal{H}\to(-\infty,+\infty]:
            v\mapsto\sup_{u\in\mathcal{H}}[\ev{u,v}_{\mathcal{H}}-f(u)]
          \end{equation}
          satisfies $f^{\ast}\in\Gamma_{0}(\mathcal{H})$.
          The function $f^{\ast}$ is called \emph{conjugate} (also named \emph{Legendre-Fenchel transform}) of $f$.
    \item \label{item:prox}
          (Proximity operator~{\cite[Def. 12.23]{bauschke_convex_2017}}).
          The \emph{proximity operator} of a function $f\in\Gamma_{0}(\mathcal{H})$ is defined by
          \begin{equation}
            \Prox_{f}:\mathcal{H}\to\mathcal{H}:x\mapsto\argmin_{y\in\mathcal{H}}
            \left[f(y)+\frac{1}{2}\norm{x-y}_{\mathcal{H}}^{2}\right].
          \end{equation}
          $f\in\Gamma_{0}(\mathcal{H})$ is said to be \emph{prox-friendly} if $\Prox_{\gamma f}$ is available as a computable operator for every $\gamma\in\mathbb{R}_{++}$.
    \item (Nonexpansive operator~\cite[Def. 4.1]{bauschke_convex_2017}).
          An operator $T:\mathcal{H}\to\mathcal{H}$ is said to be
          \emph{nonexpansive} if $(\forall x,y\in\mathcal{H})\ \norm{T(x)-T(y)}_{\mathcal{H}}\leq \norm{x-y}_{\mathcal{H}}$, in particular,
          ($\alpha$-) \emph{averaged nonexpansive} if there exists $\alpha\in(0,1)$ and a nonexpansive operator $R:\mathcal{H}\to\mathcal{H}$ such that $T=(1-\alpha)\Id + \alpha R$.
          For $f\in\Gamma_{0}(\mathcal{H})$, $\Prox_{f}$ in~\ref{item:prox} is known to be $\frac{1}{2}$-averaged nonexpansive~\cite[Prop. 12.28]{bauschke_convex_2017}.
  \end{enumerate}
\end{definition}

\begin{example}[Expression of proximity operators]
  \label{ex:proximity_operator}
  \,
  \begin{enumerate}[label=(\alph*), leftmargin=*, itemsep=1ex]
    \item \label{item:indicator}
          (Indicator function~\cite[Exm. 12.25]{bauschke_convex_2017}).
          For a nonempty closed convex subset $C\subset\mathcal{H}$, $\Prox_{\gamma\iota_{C}}$ with any $\gamma\in\mathbb{R}_{++}$ is given by the metric projection
          \begin{equation}
            P_{C}:\mathcal{H}\to\mathcal{H}:x\mapsto\argmin_{y\in\mathcal{C}}\norm{x-y}_{\mathcal{H}}.
          \end{equation}
          A closed convex set $C$ is said to be \emph{simple} if $\iota_{C}$ is prox-friendly.
    \item \label{item:conjugate}
          (Conjugate of $f$~\cite[Thm. 14.3 (ii)]{bauschke_convex_2017}).
          For $f\in\Gamma_{0}(\mathcal{H})$, $\Prox_{f^{\ast}}$ is given by
          \begin{equation}
            \Prox_{f^{\ast}}=\Id-\Prox_{f}.
          \end{equation}
    \item \label{item:shifted}
          (Shifted function~\cite[Prop. 24.8 (ii)]{bauschke_convex_2017}).
          For $f\in\Gamma_{0}(\mathcal{H})$ and $z\in\mathcal{H}$, the proximity operator of the shifted function $g:\mathcal{H}\to(-\infty,\infty]:x\mapsto f(x-z)$ is given by
          \begin{equation}
            \Prox_{g}(\cdot)=z+\Prox_{f}(\cdot-z).
          \end{equation}
  \end{enumerate}
\end{example}

\begin{fact}[Krasnosel'ski\u{\i}-Mann iteration in finite dimensional real Hilbert space~{\cite[Sec. 5.2]{bauschke_convex_2017}}]
  \label{fct:picard}
  For a nonexpansive operator $T:\mathcal{H}\to\mathcal{H}$ with $\Fix(T)\coloneqq \{u\in\mathcal{H}\mid T(u)=u\}\neq \emptyset$ and any initial point $x_{0}\in \mathcal{H}$, the sequence $(x_{k})_{k=0}^{\infty}\subset\mathcal{H}$ generated by
  \begin{equation}
    (k=0,1,2,\ldots)\quad x_{k+1}=\left[(1-\alpha_{k})\Id+\alpha_{k}T\right](x_{k})
  \end{equation}
  converges to a point in $\Fix(T)$ if $(\alpha_{k})_{k=0}^{\infty}\subset [0,1]$ satisfies $\sum_{k=0}^{\infty}\alpha_{k}(1-\alpha_{k})=\infty$. In particular, if $T$ is $\alpha$-averaged for some $\alpha\in(0,1)$, a simple Picard-type iteration
  \begin{equation}
    (k=0,1,2,\ldots)\quad x_{k+1}=T(x_{k})
  \end{equation}
  converges to a point in $\Fix(T)$.
\end{fact}

\subsection{Brief Introduction to cLiGME model}
\label{sec:cLiGME}

\begin{problem}[cLiGME model {\cite{yata_constrained_2022,yata_imposing_2024}}]\label{prob:cLiGME_model}
Let
$\mathcal{X}$,
$\mathcal{Y}$,
$\mathcal{Z}$
and
$\widetilde{\mathcal{Z}}$
be finite dimensional real Hilbert spaces.
Suppose that
(a)~$y\in\mathcal{Y}$ and $A\in \mathcal{B}(\mathcal{X}, \mathcal{Y})$;
(b)~$B\in \mathcal{B}(\mathcal{Z}, \widetilde{\mathcal{Z}})$, $\mathfrak{L}\in \mathcal{B}(\mathcal{X}, \mathcal{Z})$ and $\mu\in\mathbb{R}_{++}$;
(c)~$C(\subset\mathcal{X})$ is a nonempty simple closed convex set;
(d)~$\Psi\in \Gamma_0(\mathcal{Z})$ is
coercive, $\dom \Psi = \mathcal{Z}$ and prox-friendly.
Then
\begin{enumerate}[label=\roman*), leftmargin=0em, align=left, itemsep=1ex]
  \item
        with a tunable matrix $B\in \mathcal{B}(\mathcal{Z}, \widetilde{\mathcal{Z}})$, the \emph{Generalized Moreau Enhancement (GME)} of $\Psi$ is defined by
        \begin{equation}\label{eq:GME}
          \Psi_{B}(\cdot):= \Psi(\cdot) - \min_{v\in\mathcal{Z}}\left[
          \Psi(v) +\frac{1}{2}\bnorm{B(\cdot-v)}^2_{\widetilde{\mathcal{Z}}}
          \right];
        \end{equation}
  \item
        a \emph{constrained LiGME (cLiGME)} model is given as%
        \footnote{
          The original version of the cLiGME model~\cite{yata_constrained_2022} employs the constraint $\mathfrak{C}x\in C$, where $\mathfrak{C}$ is a linear operator.
          For simplicity, we consider the case $\mathfrak{C}=\Id$ in this paper.
        }
        \begin{equation}\label{eq:cLiGME_model}
          \text{Find~}x^{\lozenge}\in
          \mathcal{S}_{C}\coloneqq
          \argmin_{x\in C}
          J_{\Psi_{B}\circ\mathfrak{L}}(x),
        \end{equation}
        where 
        \begin{equation}\label{eq:LiGME_objective_function}
          J_{\Psi_{B}\circ\mathfrak{L}}(x)\coloneqq
          \frac{1}{2}\norm{y-Ax}_{\mathcal{Y}}^2 +
          \mu\Psi_{B}\circ \mathfrak{L}(x),
        \end{equation}
        and 
        $\Psi_{B}\circ\mathfrak{L}:\mathcal{X}\to\mathbb{R}$ is called the \emph{LiGME regularizer}~\cite{abe_linearly_2020}.
\end{enumerate}
\end{problem}

The regularizer $\Psi_{B}$ has been proposed originally in~\cite{abe_linearly_2020} as an extension of the so-called \emph{GMC penalty} in~\cite{selesnick_sparse_2017} mainly for applications to the sparsity-aware estimation.
Although $\Psi_{B}$ with $B\neq O$ is nonconvex, the convexity of the cost function $J_{\Psi_{B}\circ\mathfrak{L}}$ in~\eqref{eq:LiGME_objective_function} is achieved by a strategic tuning of \emph{GME matrix} $B$ (see, e.g.,~\cite{chen_unified_2023}) as follows:
\begin{fact}[Overall convexity condition {\cite[Prop. 1 (b)]{abe_linearly_2020}} for~\eqref{eq:LiGME_objective_function}]
  \label{fct:occ}
  Consider Problem~\ref{prob:cLiGME_model}.
  Then $J_{\Psi_{B}\circ\mathfrak{L}}$ in~\eqref{eq:LiGME_objective_function} is convex if
  \begin{equation}\label{eq:occ_general}
    A^{\ast}A-\mu\mathfrak{L}^{\ast}B^{\ast}B\mathfrak{L}\succeq O_{\mathcal{X}}.
  \end{equation}
\end{fact}

Under the overall convexity condition~\eqref{eq:occ_general} and even symmetric condition of $\Psi\in\Gamma_{0}(\mathcal{Z})$ (i.e., $\Psi\circ(-\Id)=\Psi$), the vector sequence generated by the LiGME algorithm is guaranteed to converge to a global minimizer of $J_{\Psi_{B}\circ\mathfrak{L}}$ over $x\in \mathcal{X}$~\cite[Thm. 1]{abe_linearly_2020}.
We should remark here that the even symmetric condition can be removed by relaxing~{\cite[Lemma 1]{abe_linearly_2020}} to new Lemma~\ref{lem:relative_interior} below (Note:~\cite[Lemma 1]{abe_linearly_2020} qualifies the chain rule (see Fact~\ref{fct:subdifferential}~\ref{fct:chain_rule} in \ref{appendix:known_facts}) of the subdifferential in~\cite[(D.9)]{abe_linearly_2020}).
This relaxation is necessary for Proposition~\ref{prop:cLiGME_algorithm} and Theorem~\ref{thm:proposed_algorithm} which will be keys to solve the proposed model~\eqref{eq:proposed_model} in Section~\ref{sec:proposed_regularizer} for Problem~\ref{prob:discrete}.

\begin{lemma}[A relaxation of {\cite[Lemma 1]{abe_linearly_2020}} without even symmetric condition of $\Psi$]\label{lem:relative_interior}
  If $\Psi\in\Gamma_{0}(\mathcal{Z})$ is coercive and $\dom\Psi=\mathcal{Z}$, we have for $B\in\mathcal{B}(\mathcal{Z,\widetilde{\mathcal{Z}}})$,
  \begin{equation} \label{eq:ri}
    0_{\mathcal{Z}}\in\ri\left(\dom\left(\left(\Psi+\frac{1}{2}\bnorm{B\cdot}_{\widetilde{\mathcal{Z}}}^{2}\right)^{\ast}\right)-\ran\left({B}^{\ast}\right)\right).
  \end{equation}
  \begin{proof}
    See~\ref{appendix:proof_lemma}.
  \end{proof}
\end{lemma}

The vector sequence generation~\eqref{eq:algorithm} in Proposition~\ref{prop:cLiGME_algorithm} below presents an iterative algorithm for finding a global minimizer of the cLiGME model~\eqref{eq:cLiGME_model} without assuming the even symmetry of $\Psi$.
The algorithm in Proposition~\ref{prop:cLiGME_algorithm} is designed by introducing $P_{C}$ into the original LiGME algorithm~\cite{abe_linearly_2020}.
Although such a projection involved type cLiGME algorithm is found in~\cite{kitahara_multi-contrast_2021,yata_imposing_2024} under the even symmetric condition of $\Psi$, to the best of the authors' knowledge, its explicit proof has not yet been reported.
For completeness, we present below a convergence analysis of an LiGME-type algorithm incorporating $P_{C}$ without assuming the even symmetric condition $\Psi\circ(-\Id)=\Psi$.

\begin{proposition}[A cLiGME algorithm without even symmetric condition of $\Psi$ in~\cite{abe_linearly_2020,yata_constrained_2022,yata_imposing_2024,kitahara_multi-contrast_2021} for~\eqref{eq:cLiGME_model} in Problem~\ref{prob:cLiGME_model}]\label{prop:cLiGME_algorithm}
  Consider Problem~\ref{prob:cLiGME_model} under the overall convexity condition~\eqref{eq:occ_general}.
  Assume%
  \footnote{
    $\mathcal{S}_{C}\neq\emptyset$ is guaranteed in many cases, e.g., if $C$ is compact (not limited to this case).
  }
  $\mathcal{S}_{C}\coloneqq\argmin_{x\in C}J_{\Psi_{B}\circ\mathfrak{L}}(x)\neq\emptyset$.
  Define the operator
  $T_{\cLiGME}:\mathcal{H}\coloneqq\mathcal{X}\times\mathcal{Z}\times\mathcal{Z}\to\mathcal{H}:(x, v, w)\mapsto(\xi,\zeta,\eta)$ with
  $(\sigma,\tau)\in\mathbb{R}_{++}\times\mathbb{R}_{++}$ by
  \begin{equation}\label{eq:TcLiGME}
    \hspace{-1.5em}
    \begin{cases}
      \xi=P_{C}\bigl[
      (\Id-\frac{1}{\sigma}(A^{\ast}A-\mu\mathfrak{L}^{\ast}B^{\ast}B\mathfrak{L}))x
      \\ \hfill
      -\frac{\mu}{\sigma}\mathfrak{L}^{\ast}B^{\ast}Bv
      -\frac{\mu}{\sigma}\mathfrak{L}^{\ast}w
      +\frac{1}{\sigma}A^{\ast}y
      \bigr],
      \\
      \zeta=\Prox_{\frac{\mu}{\tau}\Psi}\bigl[
      \frac{\mu}{\tau}B^{\ast}B\mathfrak{L}(2\xi-x)
      +(\Id-\frac{\mu}{\tau}B^{\ast}B)v
      \bigr],
      \\
      \eta=\Prox_{\Psi^{\ast}}\bigl[
        \mathfrak{L}(2\xi-x)+w
      \bigr]. \\
    \end{cases}
  \end{equation}
  Then the following hold:
  \begin{enumerate}[label=(\alph*), itemsep=1ex, leftmargin=1em]
    \item \label{item:fixed_point}
          The solution set $\mathcal{S}_{C}$ of Problem~\ref{prob:cLiGME_model} can be expressed as
          \begin{align}
            \label{eq:S_C}
            \mathcal{S}_{C}
            =
             & \Xi\left(\Fix(T_{\cLiGME})\right) \\
            \coloneqq
             & \left\{
            \Xi(\widebar{x},\widebar{v},\widebar{w})\in\mathcal{X}
            \mid
            (\widebar{x},\widebar{v},\widebar{w})\in\Fix(T_{\cLiGME})
            \right\},
          \end{align}
          with $\Xi:\mathcal{H}\to\mathcal{X}:(x,v,w)\mapsto x$.
    \item \label{item:nonexpansive}
          Choose $(\sigma,\tau,\kappa)\in\mathbb{R}_{++}\times\mathbb{R}_{++}\times(1,\infty)$ satisfying%
          \footnote{
            For example, \eqref{eq:stepsizes} is satisfied by any $\kappa>1$ and
            \begin{equation}
              \begin{cases}
                \sigma=\opnorm{\frac{\kappa}{2}A^{\ast}A+\mu\mathfrak{L}^{\ast}\mathfrak{L}}
                +(\kappa-1),
                \\
                \tau=(\frac{\kappa}{2}+\frac{2}{\kappa})\mu\opnorm{B}^{2}
                +(\kappa-1).
              \end{cases}
            \end{equation}
          }
          \begin{equation}\label{eq:stepsizes}
            \begin{cases}
              \sigma\Id-\frac{\kappa}{2}A^{\ast}A-\mu\mathfrak{L}^{\ast}\mathfrak{L}\succeq O_{\mathcal{X}},
              \\
              \tau\geq(\frac{\kappa}{2}+\frac{2}{\kappa})\mu\opnorm{B}^{2}.
            \end{cases}
          \end{equation}
          Then
          \begin{equation}
            \hspace{-1em}
            \mathfrak{P}\coloneqq
            \begin{pmatrix}
              \sigma\Id                  & -\mu\mathfrak{L}^{\ast}B^{\ast}B & -\mu\mathfrak{L}^{\ast}
              \\
              -\mu B^{\ast}B\mathfrak{L} & \tau\Id                          & O_{\mathcal{Z}}
              \\
              -\mu\mathfrak{L}           & O_{\mathcal{Z}}                  & \mu\Id
            \end{pmatrix}
            \succ O_{\mathcal{H}}
          \end{equation}
          and $T_{\cLiGME}$ is $\frac{\kappa}{2\kappa-1}$-averaged nonexpansive in the Hilbert space $\Hilbert{\mathcal{H}}{\mathfrak{P}}$.
    \item \label{item:convergence}
          Assume $(\sigma,\tau,\kappa)\in\mathbb{R}_{++}\times\mathbb{R}_{++}\times(1,\infty)$ satisfies~\eqref{eq:stepsizes}.
          Then for any initial point $\mathfrak{u}_{0}\coloneqq(x_{0},v_{0},w_{0})\in\mathcal{H}$, the sequence $(\mathfrak{u}_{k})_{k=0}^{\infty}\subset\mathcal{H}$ with $\mathfrak{u}_{k}\coloneqq (x_{k}, v_{k}, w_{k})$ generated by
          \begin{equation} \label{eq:algorithm}
            (k=0,1,2,\ldots)\quad\mathfrak{u}_{k+1}=T_{\cLiGME}(\mathfrak{u}_{k})
          \end{equation}
          converges to a point $(\widebar{x},\widebar{v},\widebar{w})\in\Fix(T_{\cLiGME})$ and
          \begin{equation}
            \lim_{k\to\infty}x_{k}=\widebar{x}\in\mathcal{S}_{C}.
          \end{equation}
  \end{enumerate}
  \begin{proof}
    See~\ref{appendix:proof_theorem}.
  \end{proof}
\end{proposition}

\section{Proposed LiGME regularizer with designated isolated local minimizers}
\label{sec:proposed_regularizer_and_algorithm}

\subsection{Proposed LiGME regularizer}
\label{sec:proposed_regularizer}

We propose a contrastive regularizers $\Theta_{\LiGME}$ in~\eqref{eq:theta_LiGME}, which is designed as a sum of the (shifted) GME functions~\eqref{eq:GME} of $\norm{\cdot}_{\bm{\omega}^{\ev{l}};1}$ $(l=1,2,\ldots,L)$.
By a strategic choice of $\bigl(\B^{\ev{l}}\bigr)_{l=1}^{L}$, we can make $\Theta_{\LiGME}$ fairly contrastive (see Theorem~\ref{thm:isolated} below, Fig.~\ref{fig:theta_LiGME_simple}~\subref{fig:theta_LiGME_id} and Fig.~\ref{fig:theta_LiGME_practical}).
Theorem~\ref{thm:isolated} shows that $\Theta_{\LiGME}$ with simple GME matrices has multiple isolated local minimizers at the designated vectors in $\mathfrak{A}^{N}$.
\newcommand{\MC}[1]{\left({}^{#1}\abs{\cdot}_{\text{MC}}\right)}
\begin{theorem}[Isolated local minimizers at designated points] \label{thm:isolated}
  Let
  $\mathfrak{A}\coloneqq \left\{a^{\ev{1}},a^{\ev{2}},\ldots,a^{\ev{L}}\right\}\subset \mathbb{R}$,
  and
  $\bm{\omega}^{\ev{l}} \coloneqq \left(\omega^{\ev{l}}_{1},\omega^{\ev{l}}_{2},\ldots,\omega^{\ev{l}}_{N}\right)^{\top} \in \mathbb{R}_{++}^{N}\ (l=1,2,\ldots,L)$
  satisfy
  $\sum_{l=1}^{L}\omega^{\ev{l}}_{n}=1$ $(n=1,2,\ldots,N)$.
  For
  $\Theta_{\LiGME}:\mathbb{R}^{N}\to \mathbb{R}$
  in~\eqref{eq:theta_LiGME}
  with
  $\B^{\ev{l}} \coloneqq b\I_{N}\ (l=1,2,\ldots,L)$
  and
  $b > 0$,
  the following hold:
  \begin{enumerate}[label=(\alph*), leftmargin=*, itemsep=1ex]
    \item \label{item:expression}
          $\Theta_{\LiGME}$
          has the following expression
          \begin{equation}
            \label{eq:separation}
            \myskip{0.5}
            (\x=(x_{1},x_{2},\ldots,x_{N})^{\top}\in\mathbb{R}^{N})\quad
            \Theta_{\LiGME}(\x)
            = \sum_{n=1}^{N} \psi_{n}(x_{n})
          \end{equation}
          with univariate functions
          \begin{equation}
            \psi_{n}:\mathbb{R}\to\mathbb{R}:x\mapsto  \sum_{l=1}^{L} \omega^{\ev{l}}_{n}\MC{\frac{b^{2}}{\omega^{\ev{l}}_{n}}}(x-a^{\ev{l}}),
          \end{equation}
          where the so-called \emph{Minimax Concave Penalty (MCP)}~\cite{zhang_nearly_2010}
          $\MC{\gamma}$
          with
          $\gamma > 0$
          is given by
            {
              \myskip{0.5}
              \begin{align} \label{eq:MCP}
                \hspace{-1em}\MC{\gamma}:\mathbb{R} \to \mathbb{R}: x\mapsto
                \begin{cases}
                  \abs{x} - \frac{\gamma}{2}x^{2}, & \mathrm{if}\ \abs{x}\leq \frac{1}{\gamma} \\
                  \frac{1}{2\gamma},               & \mathrm{if}\ \abs{x} > \frac{1}{\gamma},
                \end{cases}
              \end{align}
            }%
          and also given by
          $\MC{\gamma}(x)=(\abs{\cdot})_{\sqrt{\gamma}}(x)$
          with the GME
          $(\abs{\cdot})_{\sqrt{\gamma}}$
          of
          $\abs{\cdot}$
          (see, e.g.,~\cite[Prop. 12]{selesnick_sparse_2017}).
    \item \label{item:designated}
          For
          $d_{\min}\coloneqq \min\{\abs{a-a'} \mid a,a'\in \mathfrak{A}, a\neq a'\} > 0$
          and
          $\omega_{\max} \coloneqq \max\left\{\omega^{\ev{l}}_{n} \middle|\, l=1,2,\ldots,L,\ n=1,2,\ldots,N\right\} > 0$,
          let
          $b >  \sqrt{\frac{2\omega_{\max}}{d_{\min}}}(>0)$.
          Then,
          $\mathfrak{A}^{N}$ is the set of all isolated local minimizers of
          $\Theta_{\LiGME}$.
  \end{enumerate}
  \begin{proof}
    See~\ref{appendix:proof_isolated}.
  \end{proof}
\end{theorem}

Lemma~\ref{lem:LiGME_regularizer} below demonstrates that $\Theta_{\LiGME}$ in~\eqref{eq:theta_LiGME} is an LiGME regularizer.

\begin{lemma}[]
  \,
  \label{lem:LiGME_regularizer}
  \begin{enumerate}[label=(\alph*), leftmargin=0em, itemsep=1ex, align=left]
    \item \label{item:commutativity}
          (Commutativity of GME and shift operator).
          Let $\mathcal{Z}$ and $\widetilde{\mathcal{Z}}$ be finite dimensional real Hilbert spaces, $\Phi\in\Gamma_{0}(\mathcal{Z})$ be coercive with $\dom\Phi=\mathcal{Z}$, $B\in\mathcal{B}(\mathcal{Z}, \widetilde{\mathcal{Z}})$ and $z\in\mathcal{Z}$.
          Consider the GME $[\Phi(\cdot-z)]_{B}$ (see~\eqref{eq:GME}) of $\Phi(\cdot-z):\mathcal{Z}\to\mathbb{R}:u\mapsto\Phi(u-z)$.
          Then we have for $u\in\mathcal{Z}$,
          \begin{equation}\label{eq:GME_shift}
            [\Phi(\cdot-z)]_{B}(u)=\Phi_{B}(u-z).
          \end{equation}
    \item \label{item:LiGME_regularizer}
          (Reformulation of $\Theta_{\LiGME}$ in~\eqref{eq:theta_LiGME} as an LiGME regularizer).
          Define
            {
              \myskip{0.5}
              \begin{align}
                \hspace{-2em}
                \mathfrak{L}: &
                \mathbb{R}^{N}\to
                \bigl(\mybigtimes_{l=1}^{L}\mathbb{R}^{N}\bigr)
                :\uu\mapsto (\uu)_{l=1}^{L},                                                                                                             \\
                \hspace{-2em}
                \Psi:         &
                \bigl(\mybigtimes_{l=1}^{L}\mathbb{R}^{N}\bigr)
                \to\mathbb{R}:\bigl(\uu^{\ev{l}}\bigr)_{l=1}^{L}\mapsto\sum_{l=1}^{L} \bnorm{\uu^{\ev{l}}-a^{\ev{l}}\ones_{N}}_{\bm{\omega}^{\ev{l}};1}, \\
                \hspace{-2em}
                B:           &
                \bigl(\mybigtimes_{l=1}^{L}\mathbb{R}^{N}\bigr)
                \to
                \bigl(\mybigtimes_{l=1}^{L}\mathbb{R}^{q^{\ev{l}}}\bigr)
                :\bigl(\uu^{\ev{l}}\bigr)_{l=1}^{L}\mapsto \bigl(\B^{\ev{l}}\uu^{\ev{l}}\bigr)_{l=1}^{L}.
              \end{align}}%
          Then, $\Theta_{\LiGME}$ in~\eqref{eq:theta_LiGME} can be expressed as a special instance of the LiGME regularizer in~\eqref{eq:LiGME_objective_function}, i.e.,
          \begin{equation}
            \Theta_{\LiGME}
            =\Psi_{B}\circ\mathfrak{L},
          \end{equation}
          where $\Psi$ is coercive with $\dom(\Psi)=\mathcal{X}$.
  \end{enumerate}

  \begin{proof}
    See~\ref{appendix:proof_LiGME_regularizer}.
  \end{proof}
\end{lemma}

\subsection{Proposed model and algorithm for discrete-valued signal estimation}
\label{sec:proposed_algorithm}

\begin{algorithm}[t]
  \caption{A cLiGME algorithm for proposed model~\eqref{eq:proposed_model}}
  \begin{algorithmic} [1]
    \label{alg}
    \STATE{
    Choose
    $\left(\x_{0},(\v_{0}^{\ev{l}})_{l=1}^{L},(\w_{0}^{\ev{l}})_{l=1}^{L}\right)\in\mathbb{R}^{N}
      \times
      \left(\mybigtimes_{l=1}^{L}\mathbb{R}^{N}\right)
      \times
      \left(\mybigtimes_{l=1}^{L}\mathbb{R}^{N}\right)$.
    }
    \STATE{
      Choose
      $(\sigma, \tau, \kappa)\in\mathbb{R}_{++}\times\mathbb{R}_{++}\times(1,\infty)$ satisfying\footnotemark
      \begin{equation}
        \label{eq:stepsizes2}
        \begin{cases}
          (\sigma-\mu L)\Id-\frac{\kappa}{2}\A^{\top}\A\succeq \O_{N\times N},
          \\
          \tau\geq(\frac{\kappa}{2}+\frac{2}{\kappa})\mu\max\left\{
          \opnorm{\B^{\ev{l}}}^2\mid l=1,2,\ldots,L
          \right\}.
        \end{cases}
      \end{equation}
    }
    \FOR{$k=0,1,2,\ldots$}
    \STATE{
      \COMMENT{
        Modification in \ref{appendix:modification} can be inserted if necessary.
      }
    }
    \STATE{
      $\x_{k+1} \gets P_{\mathcal{C}}\Bigl[
          \x_{k}-\frac{1}{\sigma}\A^{\top}(\A\x_{k}-\y)$
            \\ \hfill
          $+\frac{\mu}{\sigma}\sum_{l=1}^{L}\left({\B^{\ev{l}}}^{\top}\B^{\ev{l}}\left(\x_{k}-\v_{k}^{\ev{l}}\right)-\w_{k}^{\ev{l}}\right)
          \Bigr]$
    }
    \FOR{$l=1,2,\ldots,L$}
    \STATE{
    $\v_{k+1}^{\ev{l}} \gets a^{\ev{l}}\ones_{N}+\Prox_{\frac{\mu}{\tau}\norm{\cdot}_{\bm{\omega}^{\ev{l}};1}}
      \Bigl[\frac{\mu}{\tau}{\B^{\ev{l}}}^{\top}\B^{\ev{l}}\left(2\x_{k+1}-\x_{k}-\v_{k}^{\ev{l}}\right)$
    \\ \hfill
    $+\v_{k}^{\ev{l}}-a^{\ev{l}}\ones_{N}\Bigr]$
    }
    \STATE{
    $\w_{k+1}^{\ev{l}} \gets \left(\Id - \Prox_{\norm{\cdot}_{\bm{\omega}^{\ev{l}};1}}\right)
      \Bigl[2\x_{k+1}-\x_{k}+\w_{k}^{\ev{l}}- a^{\ev{l}}\ones_{N}\Bigr]$
    }
    \ENDFOR
    \ENDFOR
    \RETURN{$\x_{k+1}$}
  \end{algorithmic}
\end{algorithm}
\footnotetext{
  \label{foot:sigma_tau_kappa}
  For example, \eqref{eq:stepsizes2} is satisfied by any $\kappa>1$ and
  \begin{equation}
    \hspace{-0.5em}
    \begin{cases}
      \sigma\coloneqq\frac{\kappa}{2}\opnorm{\A}^2+\mu L+(\kappa-1),
      \\
      \tau\coloneqq({\frac{\kappa}{2}+\frac{2}{\kappa}})\mu
      \max\left\{\opnorm{\B^{\ev{l}}}^2\mid l=1,2,\ldots,L\right\}+(\kappa-1).
    \end{cases}
  \end{equation}
}

By using $\Theta_{\LiGME}$ (see~\eqref{eq:theta_LiGME}) in Step~1 of Scheme~\ref{scheme}, we propose the following model:
\begin{equation}\label{eq:proposed_model}
  \myskip{0.0}
  \hspace{-2em}
  \text{Find~}\x^{\lozenge}\in
  \argmin_{\x\in\mathcal{C}}
  \overbrace{
  \frac12\norm{\y-\A\x}_2^2+\mu
  \underbrace{
  \sum_{l=1}^{L}\bigl(\norm{\cdot}_{\bm{\omega}^{\ev{l}};1}\bigr)_{\B^{\ev{l}}}\bigl(\x-a^{\ev{l}}\ones_{N}\bigr)
  }_{=\Theta_{\LiGME}(\x)}
  }^{J_{\Theta_{\LiGME}}(\x)=},
\end{equation}
where $\mathcal{C}\subset\mathbb{R}^{N}$ is a closed convex superset of $\mathfrak{A}^{N}$, and $\B^{\ev{l}}\in\mathbb{R}^{q^{\ev{l}}\times N}$ $(l=1,2,\ldots,L)$ are GME matrices (Note: $q^{\ev{l}}$ can be any positive integer).
If we set $\B^{\ev{l}}=\O_{q^{\ev{l}}\times N}$ $(l=1,2,\ldots,L)$, then $\Theta_{\LiGME}=\Theta_{\SOAV}$ holds true.
Under the setting $(\mathfrak{L}, \Psi, B)$ in Lemma~\ref{lem:LiGME_regularizer}~\ref{item:LiGME_regularizer}, the proposed model~\eqref{eq:proposed_model} can be seen as a special instance of the cLiGME model in Problem~\ref{prob:cLiGME_model} although $\Psi$ violates $\Psi\circ(-\Id) = \Psi$.

By tuning properly $\bigl(\B^{\ev{l}}\bigr)_{l=1}^{L}$ in the proposed LiGME regularizers~\eqref{eq:theta_LiGME}, we can make the proposed nonconvexly-regularized least squares model~\eqref{eq:proposed_model} convex.
More precisely, by noting that the proposed model~\eqref{eq:proposed_model} is an instance of Problem~\ref{prob:cLiGME_model}, Fact~\ref{fct:occ} and Lemma~\ref{lem:LiGME_regularizer}~\ref{item:LiGME_regularizer} ensure that $J_{\Theta_{\LiGME}}$ in~\eqref{eq:proposed_model} is convex if $\bigl(\B^{\ev{l}}\bigr)_{l=1}^{L}$ satisfies
\begin{equation}\label{eq:occ_special}
  \A^{\top}\A-\mu\sum_{l=1}^{L}{\B^{\ev{l}}}^{\top}\B^{\ev{l}}\succeq \O_{N\times N}.
\end{equation}
For example, the following $\B^{\ev{l}}$ satisfies~\eqref{eq:occ_special}:
\begin{equation} \label{eq:choice_Bl}
  \B^{\ev{l}}\coloneqq\sqrt{{\gamma^{\ev{l}}}/{\mu}}\,\A\in\mathbb{R}^{M\times N}\
  (l=1,2,\ldots,L),
\end{equation}
where $\gamma^{\ev{l}}\in\mathbb{R}_{+}$ $(l=1,2,\ldots,L)$ are chosen to satisfy $\sum_{l=1}^{L}\gamma^{\ev{l}}\leq 1$
(Note: $\B^{\ev{l}}$ was introduced originally in~\cite{selesnick_sparse_2017} for nonconvex enhancement of LASSO regularizer~\cite{tibshirani_regression_1996}, i.e., $L=1$ case).

Algorithm~\ref{alg} illustrates a concrete expression of the Krasnosel'ski\u{\i}-Mann iteration~\eqref{eq:TcLiGME},\eqref{eq:stepsizes} and~\eqref{eq:algorithm} in Proposition~\ref{prop:cLiGME_algorithm} in Section~\ref{sec:cLiGME},
where we used the expressions, in Example~\ref{ex:proximity_operator}~\ref{item:conjugate} and \ref{item:shifted}, of the proximity operators.
Theorem~\ref{thm:proposed_algorithm} below presents a convergence guarantee of Algorithm~\ref{alg}.
We remark that $\norm{\cdot}_{\bm{\omega}^{\ev{l}};1}$ in~\eqref{eq:theta_SOAV} and Algorithm~\ref{alg} is prox-friendly%
\footnote{
$\Prox_{\gamma\norm{\cdot}_{\bm{\omega}^{\ev{l}};1}}:\mathbb{R}^{N}\to\mathbb{R}^{N}:\uu\coloneqq(u_{1},u_{2},\ldots u_{N})^{\top}\mapsto\p\coloneqq(p_{1},p_{2},\ldots p_{N})^{\top}$
$(\gamma\in\mathbb{R_{++}})$
is given by
$  (n=1,2,\ldots,N)\quad p_{n}=\sgn(u_{n})\max\left\{0,\abs{u_{n}}-\gamma\omega^{\ev{l}}_{n}\right\}
$~\cite[Prop. 2]{kowalski_sparse_2009}.
}.

\begin{theorem}[Convergence property of Algorithm~\ref{alg}]
  \label{thm:proposed_algorithm}
  Consider the proposed model~\eqref{eq:proposed_model}.
  Suppose that $\bigl(\B^{\ev{l}}\bigr)_{l=1}^{L}$ satisfies the overall convexity condition~\eqref{eq:occ_special} of $J_{\Theta_{\LiGME}}$.
  Then a sequence $(\x_{k})_{k=0}^{\infty}\subset \mathbb{R}^{N}$ generated by Algorithm~\ref{alg} converges to a global minimizer of $J_{\Theta_{\LiGME}}$ over $\mathcal{C}$.
  \begin{proof}
    Combining Proposition~\ref{prop:cLiGME_algorithm} and Lemma~\ref{lem:LiGME_regularizer}~\ref{item:LiGME_regularizer} completes the proof of Theorem~\ref{thm:proposed_algorithm}.
  \end{proof}
\end{theorem}

To enhance the accuracy in a discrete-valued signal estimation by Algorithm~\ref{alg}, we also propose a pair of simple modifications, in \ref{appendix:modification}, which are inserted into line~4 of Algorithm~\ref{alg}.
Both modifications exploit adaptively the discrete information regarding $\mathfrak{A}^{N}$.
Although Algorithm~\ref{alg} with such modifications has no guarantee of convergence to a global minimizer of $J_{\Theta_{\LiGME}}$ over $\mathcal{C}$, numerical experiments in Section~\ref{sec:numerical_experiment} demonstrate such modifications improve the estimation accuracy.

\section{Numerical experiments}
\label{sec:numerical_experiment}

\begin{figure*}[t]
  \centering
  \begin{minipage}[b]{2\columnwidth}
    \begin{minipage}[b]{0.32\linewidth}
      \centering
      \includegraphics[width=0.9\linewidth]{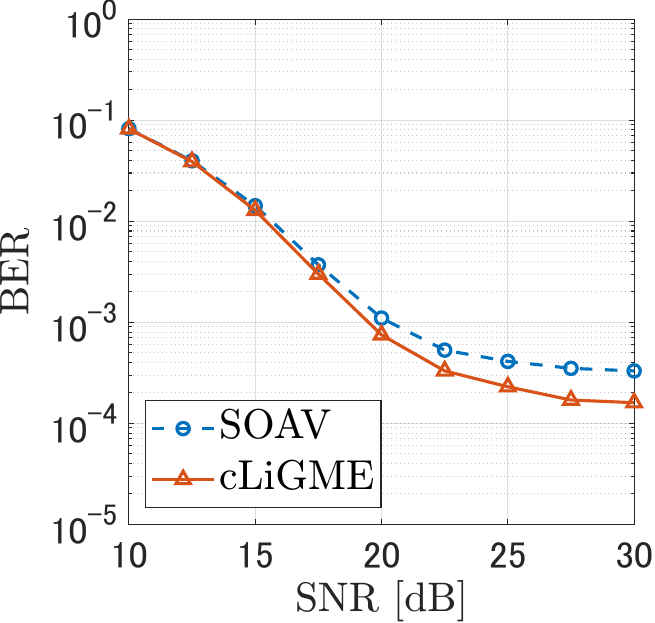}
      \subcaption{4-QAM ($N=50$, $M=35$)}
      \label{fig:4-QAM}
    \end{minipage}
    \hfill
    \begin{minipage}[b]{0.32\linewidth}
      \centering
      \includegraphics[width=0.9\linewidth]{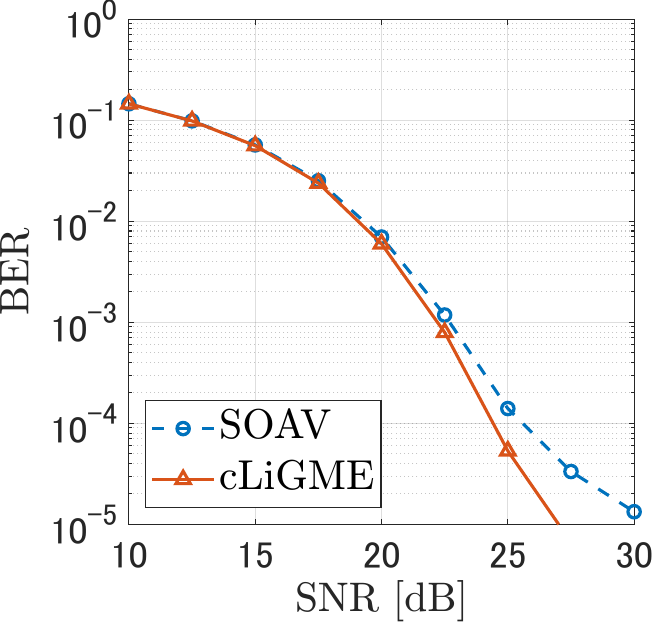}
      \subcaption{8-PSK ($N=50$, $M=45$)}
      \label{fig:8-PSK}
    \end{minipage}
    \hfill
    \begin{minipage}[b]{0.32\linewidth}
      \centering
      \includegraphics[width=0.9\linewidth]{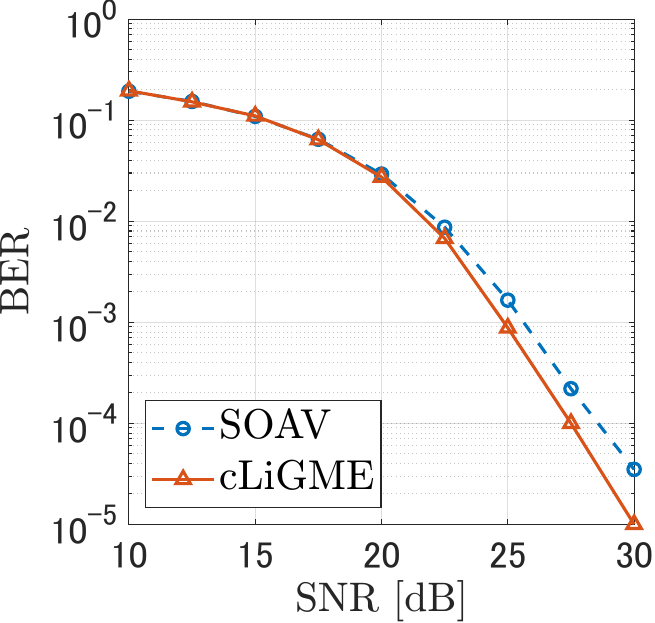}
      \subcaption{16-QAM ($N=50$, $M=50$)}
      \label{fig:16-QAM}
    \end{minipage}
    \caption{BER vs SNR (4-QAM, 8-PSK, 16-QAM)}
    \label{fig:nonconvex_enhancement}
  \end{minipage}
  \begin{minipage}[b]{2\columnwidth}
    \begin{minipage}[b]{0.66\linewidth}
      \begin{minipage}[b]{0.48\linewidth}
        \centering
        \includegraphics[width=0.9\linewidth]{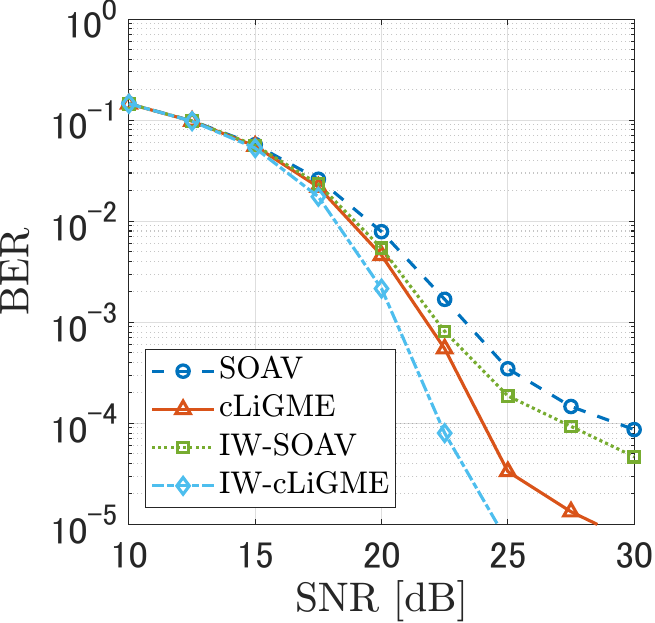}
        \subcaption{Iterative reweighting}
        \label{fig:iw}
      \end{minipage}
      \hfill
      \begin{minipage}[b]{0.48\linewidth}
        \centering
        \includegraphics[width=0.9\linewidth]{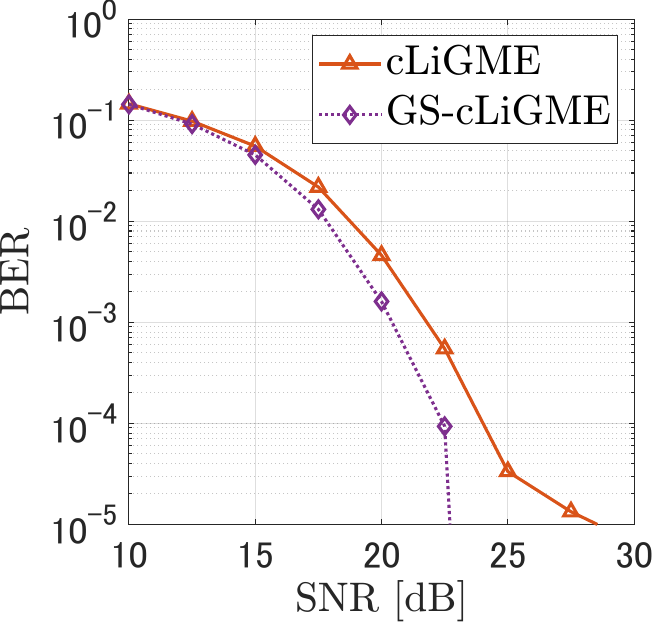}
        \subcaption{Generalized superiorization}
        \label{fig:sup}
      \end{minipage}
      \caption{BER vs SNR (8-PSK, $N=50$, $M=45$)}
      \label{fig:ber_performances}
    \end{minipage}
    \hfill
    \begin{minipage}[b]{0.285\linewidth}
      \centering
      \includegraphics[width=\linewidth]{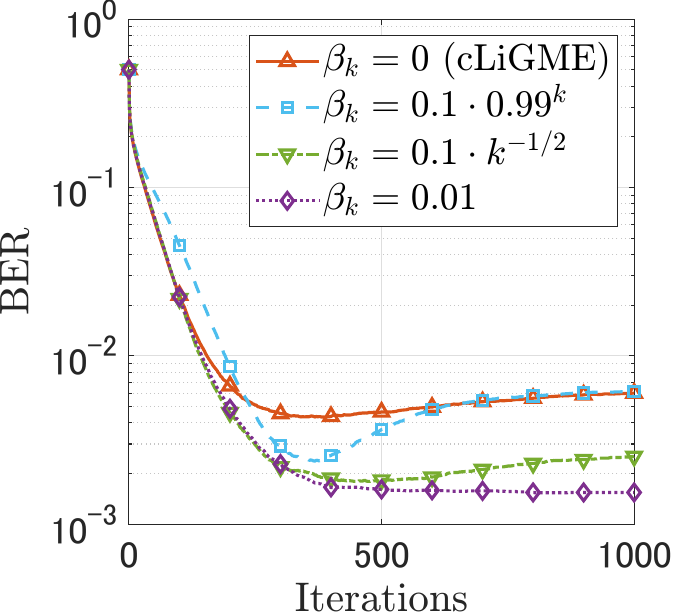}
      \vspace{0.01em} %
      \caption{BER vs iterations \\(generalized superiorization)}
      \label{fig:perturbations}
    \end{minipage}
  \end{minipage}
\end{figure*}

\subsection{An application of the proposed model to MIMO signal detection problem}
\label{sec:application}

We conducted numerical experiments in a scenario of MIMO signal detection~\cite{chen_manifold_2017} with $N$ transmit antennas and $M$ receive antenna.
MIMO signal detection has been formulated as an estimation problem of the transmitted signal $\xstar\in\mathbb{C}^{N}$ from the received signal $\y=\A\xstar+\eps\in\mathbb{C}^{M}$, where
(i)~$\xstar\in\mathfrak{A}^{N}\coloneqq\left\{a^{\ev{1}},a^{\ev{2}},\ldots,a^{\ev{L}}\right\}^{N}\subset\mathbb{C}^{N}$,
(ii)~$\A\in\mathbb{C}^{M\times N}$ is a channel matrix, and
(iii)~$\eps\in\mathbb{C}^{M}$ is noise (Note: MIMO systems with $M<N$, such as Fig.~\ref{fig:ber_performances}~\subref{fig:4-QAM} and~\subref{fig:8-PSK}, are known as overloaded (undetermined) MIMO~\cite{wong_efficient_2007,hayakawa_convex_2017}).
Via a simple $\mathbb{C}\rightleftarrows\mathbb{R}^{2}$ translation (see, e.g.,~\cite{mizoguchi_hypercomplex_2019}), the complex-valued system can be translated into a real-valued system $\widehat{\y}=\widetilde{\A}\widehat{\xstar}+\widehat{\eps}$, where
\begin{equation}\label{eq:translation}
  \hspace{-0.5em}
  \begin{aligned}
     &
    \widehat{(\cdot)}:\mathbb{C}^{d}\to\mathbb{R}^{2d}:\uu\mapsto\widehat{\uu}\coloneqq\begin{pmatrix}\Re(\uu)\\ \Im(\uu)\end{pmatrix}\ (d=N,M), \\
     &
    \widetilde{\A}\coloneqq
    \begin{pmatrix}
      \Re(\A) & -\Im(\A) \\
      \Im(\A) & \Re(\A)
    \end{pmatrix}\in\mathbb{R}^{2M\times 2N},                                                                                                    \\
     &
    \widehat{\xstar}\in\widehat{\mathfrak{A}^{N}}\coloneqq\left\{\widehat{\dfrak}\in\mathbb{R}^{2N}\middle|\, \dfrak\in\mathfrak{A}^{N}\subset\mathbb{C}^{N}\right\}.
  \end{aligned}
\end{equation}
For the estimation of $\widehat{\xstar}\in\widehat{\mathfrak{A}^{N}}\subset\mathbb{R}^{2N}$, the model~\eqref{eq:regularized_least_squares} in Step~1 of Scheme~\ref{scheme} can be rewritten as:
\begin{equation}\label{eq:proposed_model_MIMO}
  \text{Find }\widehat{\x^{\lozenge}}\in\argmin_{\widehat{\x}\in\mathcal{C}}J_{\Theta}\bigl(\widehat{\x}\bigr)\coloneqq\frac{1}{2}
  \bnorm{\widehat{\y}-\widetilde{\A}\widehat{\x}}_{2}^{2}
  +\mu\Theta\bigl(\widehat{\x}\bigr),
\end{equation}
where $\Theta:\mathbb{R}^{2N}\to\mathbb{R}$ is a regularizer, and $\mathcal{C}\subset\mathbb{R}^{2N}$ is a closed convex superset of $\widehat{\mathfrak{A}^{N}}$ as an approximation of $\widehat{\mathfrak{A}^{N}}$.
In this experiment, we employed $\mathcal{C}\coloneqq \conv(\widehat{\mathfrak{A}^{N}})$ as the constraint set $\mathcal{C}$
because 
(i) $\conv(\widehat{\mathfrak{A}^{N}})\subset\mathbb{R}^{2N}$ is the smallest closed convex superset of $\widehat{\mathfrak{A}^{N}}$, i.e., the best outer convex approximation of $\widehat{\mathfrak{A}^{N}}$, and 
(ii) the metric projection onto $\conv(\widehat{\mathfrak{A}^{N}})\subset\mathbb{R}^{2N}$ has closed-form expressions in the following all cases (see Remark~\ref{remark:computability}).

In a case where digital QAM (quadrature amplitude modulation) is employed, the finite alphabet (called the constellation set in the field of MIMO) $\mathfrak{A}$ is usually expressed as $\mathfrak{A} = \mathfrak{A}_{\qam} + j \mathfrak{A}_{\qam}$, where $\mathfrak{A}_{\qam}\coloneqq\left\{a_{\qam}^{\ev{1}},a_{\qam}^{\ev{2}},\ldots,a_{\qam}^{\ev{L_{\qam}}}\right\}\subset\mathbb{R}$ is a finite set (e.g., $\mathfrak{A}_{\qam}=\{-1,1\}$ for 4-QAM and $\mathfrak{A}_{\qam}=\{\pm1,\pm3\}$ for 16-QAM), and $L_{\qam}^{2}=L$.
In this case, for $\x \in \mathbb{C}^{N} (\Leftrightarrow \widehat{\x} \in \mathbb{R}^{2N})$, we have the relation 
$\x \in \mathfrak{A}^{N} \Leftrightarrow \widehat{\x} \in \mathfrak{A}_{\qam}^{2N}$.
From this observation, the proposed LiGME regularizer~\eqref{eq:theta_LiGME} can be rewritten as
\begin{equation}\label{eq:theta_LiGME_QAM}
  \Theta_{\LiGME}^{\qam}\bigl(\widehat{\x}\bigr)\coloneqq
  \sum_{l=1}^{L_{\qam}}
  \bigl(\norm{\cdot}_{\bm{\omega}^{\ev{l}};1}\bigr)_{\B^{\ev{l}}}
  \bigl(\widehat{\x}-a_{\qam}^{\ev{l}}\ones_{2N}\bigr),
\end{equation}
where, for each $l=1,2,\ldots,L_{\qam}$, $\bm{\omega}^{\ev{l}}\in\mathbb{R}_{++}^{2N}$ is a weighting vector, and $\B^{\ev{l}}$ is a tuning matrix satisfying the overall convexity condition~\eqref{eq:occ_special}.

\def\vec{\text{vec}}
In a case where the PSK (phase shift keying) is employed, the finite alphabet $\mathfrak{A}\subset\mathbb{C}$ is expressed as $\mathfrak{A}\coloneqq\mathfrak{A}_{\psk}\coloneqq\left\{ a^{\ev{l}}\coloneqq\exp[j(l-1)2\pi/L]\mid l=1,2,\ldots,L\right\}$, where $L=8$ is used for 8-PSK.
In this case, we use
\begin{equation}\label{eq:theta_LiGME_PSK}
  \hspace{-2em}
  \begin{aligned}
    \Theta_{\LiGME}^{\text{\psk}}\bigl(\widehat{\x}\bigr)
    \coloneqq &
    \sum_{l=1}^{L}
    \bigl(\norm{\cdot}_{\bm{\omega}^{\ev{l}};2,1}\bigr)_{\B^{\ev{l}}}
    \left(
    \widehat{\x}-\widehat{(a^{\ev{l}}\ones_{N})}
    \right)                                                \\
    =         &
    \sum_{l=1}^{L}
    \bigl(\norm{\cdot}_{\bm{\omega}^{\ev{l}};2,1}\bigr)_{\B^{\ev{l}}}
    \left(\widehat{\x}-
    \begin{pmatrix}
      \Re(a^{\ev{l}})\ones_{N} \\
      \Im(a^{\ev{l}})\ones_{N}
    \end{pmatrix}\right) \\
  \end{aligned}
\end{equation}
as the regularizer $\Theta$ in~\eqref{eq:proposed_model_MIMO},
where, for each $l=1,2,\ldots,L$ and $\uu\coloneqq(u_{1},u_{2},\ldots,u_{N})^{\top}\in\mathbb{C}^{N}$,
$\bnorm{\widehat{\uu}}_{\bm{\omega}^{\ev{l}};2,1}
  \coloneqq\sum_{n=1}^{N}\omega^{\ev{l}}_{n}\norm{(\Re(u_{n}),\Im(u_{n}))}_{2}
  \left(=\sum_{n=1}^{N}\omega^{\ev{l}}_{n}\abs{u_{n}}\right)$
is a weighted $\ell_{2,1}$-norm associated with a weighting vector $\bm{\omega}^{\ev{l}}\in\mathbb{R}_{++}^{N}$, 
and $\B^{\ev{l}}$ is a tuning matrix satisfying the overall convexity condition~\eqref{eq:occ_special}.
We note that Algorithm~\ref{alg} can be applied even to the minimization of $J_{\Theta_{\LiGME}^{\text{\psk}}}$ over a closed convex set $\mathcal{C}$ by using%
\footnote{
  For $\uu\coloneqq(u_{1},u_{2},\ldots,u_{N})\in\mathbb{C}^{N}$, 
$
  \Prox_{\gamma\norm{\cdot}_{\bm{\omega}^{\ev{l}};2,1}}:\mathbb{R}^{2N}\to\mathbb{R}^{2N}:\widehat{\uu}\mapsto \p\coloneqq(p_{1},p_{2},\ldots p_{2N})^{\top}$ $(\gamma\in\mathbb{R}_{++})$ is given by
  $
  \myskip{0.2}
  (n=1,2,\ldots,N)\quad
  \begin{pmatrix}
    p_{n} \\
    p_{N+n}
  \end{pmatrix}
  =\max\left\{0,1-\frac{\gamma\omega^{\ev{l}}_{n}}{\abs{u_{n}}}\right\}
  \begin{pmatrix}
    \Re(u_{n}) \\
    \Im(u_{n})
  \end{pmatrix}$~\cite[Prop. 2]{kowalski_sparse_2009}.
} 
$\norm{\cdot}_{\bm{\omega}^{\ev{l}};2,1}:\mathbb{R}^{2N}\to\mathbb{R}_{+}$ and $\widehat{(a^{\ev{l}}\ones_{N})}\in \mathbb{R}^{2N}$ $(l=1,2,\ldots,L)$ in place of $\norm{\cdot}_{\bm{\omega}^{\ev{l}};1}:\mathbb{R}^{N}\to\mathbb{R}_{+}$ and $a^{\ev{l}}\ones_{N}\in\mathbb{R}^{N}$, respectively.

\begin{remark}[Computability of $P_{\conv(\widehat{\mathfrak{A}^{N}})}$]\label{remark:computability}
  \,

  \begin{enumerate}[label=(\alph*), leftmargin=*, align=left]
    \item For QAM, $\conv(\widehat{\mathfrak{A}^{N}})=[\min(\mathfrak{A}_{\qam}), \max(\mathfrak{A}_{\qam})]^{2N}$ holds, i.e., $\conv(\widehat{\mathfrak{A}^{N}})$ is a box constraint.
    By~\cite[Prop. 29.3 and Exm. 29.21]{bauschke_convex_2017}, $P_{\conv(\widehat{\mathfrak{A}^{N}})}$ has a closed-form expression.
    \item For 8-PSK, $\conv(\widehat{\mathfrak{A}^{N}})=\mathcal{P}((\conv(\widehat{\mathfrak{A}}))^{N})$ holds with a proper permutation $\mathcal{P}:\mathbb{R}^{2N}\to\mathbb{R}^{2N}$, where the projection $P_{\conv(\widehat{\mathfrak{A}})}:\mathbb{R}^{2}\to\conv(\widehat{\mathfrak{A}})(\subset\mathbb{R}^{2})$ onto the regular octagon $\conv(\widehat{\mathfrak{A}})$ has a closed-form expression (see, e.g., \cite{castaneda_vlsi_2018}).
    By~\cite[Prop. 29.3]{bauschke_convex_2017}, $P_{\conv(\widehat{\mathfrak{A}^{N}})}$ has a closed-form expression.
  \end{enumerate}
\end{remark}

\subsection{Experimental setting and results on MIMO detection}
\label{sec:experimental_setting}

\def\4qam{\text{4-qam}}
\def\8psk{\text{8-psk}}
\def\16qam{\text{16-qam}}

We tested the following three modulations:
(a)~4-QAM,
(b)~8-PSK,
and
(c)~16-QAM.
In this experiment according to~\cite{chen_manifold_2017}, we chose randomly
(i)~$\xstar\in\mathfrak{A}^{N}$,
(ii)~$\A\coloneqq\sqrt{\mathbf{R}}\mathbf{G}\in\mathbb{C}^{M\times N}$, where each entry of $\mathbf{G}\in\mathbb{C}^{M\times N}$ was sampled from the complex gaussian distribution $\mathcal{N}_{\mathbb{C}}(0,1/M)$, and $\mathbf{R}\in\mathbb{R}^{M\times M}$ satisfies $(\mathbf{R})_{r,c}=0.5^{\abs{r-c}}$ $(r,c=1,2,\ldots,M)$, and
(iii)~each entry of $\eps\in\mathbb{C}^{M}$ was sampled from $\mathcal{N}_{\mathbb{C}}(0,\sigma_{\varepsilon}^2)$ with a variance $\sigma_{\varepsilon}^2\in\mathbb{R}_{++}$, which was chosen so that $10\log_{10}\frac{E[\norm{\xstar}^{2}]}{\sigma_{\varepsilon}^{2}}$ (dB) achieves a given SNR (signal-to-noise ratio).

In Step~1 of Scheme~\ref{scheme}, we considered the model~\eqref{eq:proposed_model} with
(i)~$\Theta_{\LiGME}^{\qam}$ in~\eqref{eq:theta_LiGME_QAM} for QAM and
(ii)~$\Theta_{\LiGME}^{\psk}$ in~\eqref{eq:theta_LiGME_PSK} for PSK,
where the convex hull $\conv(\widehat{\mathfrak{A}^{N}})$ of $\widehat{\mathfrak{A}^{N}}$ was employed as $\mathcal{C}$.
In this experiment, we compared numerical performance of
(i)~the proposed cLiGME model~\eqref{eq:proposed_model} with that of
(ii)~the SOAV model \cite{hayakawa_convex_2017} as convex models, for discrete-valued signal estimation, whose global minimizers can be approximated iteratively by a certain algorithm. 
The SOAV model is a special instance of the proposed model~\eqref{eq:proposed_model} with $\B^{\ev{l}}=\O_{q^{\ev{l}}\times 2N}$ $(l=1,2,\ldots,L_{\clubsuit})$ in the regularizers $\Theta_{\LiGME}^{\qam}$ and $\Theta_{\LiGME}^{\psk}$, where $L_{\clubsuit}=2$ (4-QAM), $8$ (8-PSK) and $4$ (16-QAM).
For the cLiGME model, we used Algorithm~\ref{alg} (denoted by `cLiGME') by employing the following tuning matrices in~\eqref{eq:proposed_model}
\begin{equation}\label{eq:Bl}
  \hspace{-1em}
  \B^{\ev{l}}=\sqrt{0.99/\mu L_{\clubsuit}}\,\widetilde{\A}\in\mathbb{R} ^{2M\times 2N}\,(l=1,2,\ldots,L_{\clubsuit})
\end{equation}
to achieve the overall convexity condition~\eqref{eq:occ_special} according to~\eqref{eq:choice_Bl}, where $\widetilde{\A}$ is obtained via $\mathbb{C}\rightleftarrows\mathbb{R}^{2}$ translation (see~\eqref{eq:translation} in Section~\ref{sec:application}), and
$\mu$ is a predetermined regularization parameter.
Since SOAV model can be reduced to the cLiGME model~\eqref{eq:proposed_model} with $\B^{\ev{l}}=\O_{2M\times 2N}$ $(l=1,2,\ldots,L_{\clubsuit})$, we used Algorithm~\ref{alg} (denoted by `SOAV') with $\B^{\ev{l}}=\O_{2M\times 2N}$ $(l=1,2,\ldots,L_{\clubsuit})$ for the SOAV model.
For both `cLiGME' and `SOAV', we employed the same
(i)~stepsize $(\sigma,\tau,\kappa)\in\mathbb{R}_{++}\times\mathbb{R}_{++}\times(1,\infty)$ as footnote~\ref{foot:sigma_tau_kappa} in Algorithm~\ref{alg} with $\kappa=1.001$, and
(ii)~initial points $\x_{0}=\bm{0}$, $\v^{\ev{l}}_{0}=\bm{0}$ and $\w^{\ev{l}}_{0}=\bm{0}$ $(l=1,2,\ldots, L_{\clubsuit})$.
Since $\mathcal{C}\coloneqq\conv(\widehat{\mathfrak{A}^{N}})\subset\mathbb{R}^{2N}$ is compact, `cLiGME' and `SOAV' can find their global minimizers, respectively (see Theorem~\ref{thm:proposed_algorithm}).

Before evaluating numerical performance, let us examine contrastiveness of $\Theta_{\LiGME}$ used in these experiments.
Fig.~\ref{fig:theta_LiGME_practical} shows the function values of $\Theta_{\LiGME}^{\psk}$ in~\eqref{eq:theta_LiGME_PSK} designed with $\B^{\ev{l}}$ in~\eqref{eq:Bl} hence achieving the overall convexity condition~\eqref{eq:occ_special}.
Each numerical value of $\Theta_{\LiGME}(\x)$ is computed with ISTA-type algorithm \cite{combettes_signal_2005} (Note: the function value of $\Theta_{\LiGME}$ is not required in the proposed Algorithm~\ref{alg}).
As seen from Fig.~\ref{fig:theta_LiGME_practical}~\subref{fig:theta_LiGME_practical_kakudai}, we observe numerically that $\Theta_{\LiGME}(\x)$ is certainly contrastive as a regularizer for discrete-valued signal estimation because each constellation point in $\mathfrak{A}$ is located at an isolated local minimizer of $\Theta_{\LiGME}(\x)$ as we expected.

As a performance metric, we adopted averaged BER%
\footnote{
  BER was computed by a MATLAB code with `qamdemod', `pskdemod' and `biterr'.
} 
(bit error rate) over 1,000 independent realizations of $(\x^\star,\A,\eps)$.
The parameter $\mu$ was chosen to achieve the lowest BER from the set $\{10^{i}\mid i=-6,-5\ldots,1\}$ at each SNR.

Fig.~\ref{fig:nonconvex_enhancement} shows BER of `SOAV' and `cLiGME' at each SNR, where
Algorithm~\ref{alg} were terminated when the iteration number $k$ exceeded 1,000,
$\omega^{\ev{l}}_{n}=1/L_{\clubsuit}$ $(l=1,2,\ldots,L_{\clubsuit};\,n=1,2,\ldots,N_{\clubsuit})$ in~\eqref{eq:proposed_model} were fixed ($N_{\clubsuit}=2N$ (4-QAM, 16-QAM) and $N$ (8-PSK)).
From Fig.~\ref{fig:nonconvex_enhancement}~\subref{fig:4-QAM}, \subref{fig:8-PSK} and \subref{fig:16-QAM}, 
we can see that  
(i)~both `SOAV' and `cLiGME' achieve lower BER as SNR becomes higher, and
(ii)~compared with `SOAV', `cLiGME' improves BER largely at high SNR.
These observations show the effectiveness of the proposed nonconvex $\Theta_{\LiGME}$ compared with the convex regularizer $\Theta_{\SOAV}$.

In the following, we confirmed the further performance improvements of `cLiGME' by the proposed
(i)~iterative reweighting and
(ii)~generalized superiorization
in 8-PSK (see \ref{appendix:modification}), where we terminated Algorithm~\ref{alg} if the iteration number $k$ exceeded 500, considering realistic time constraints.

Fig.~\ref{fig:ber_performances}~\subref{fig:iw} shows BER, at each SNR, of
(i)~`SOAV',
(ii)~`cLiGME',
(iii)~`SOAV' with iterative reweighting in~\eqref{eq:iw} in \ref{appendix:modification} (denoted by `IW-SOAV'), and
(iv)~`cLiGME' with iterative reweighting in~\eqref{eq:iw} (denoted by `IW-cLiGME'),
where the machine epsilon was used as $\delta$ in~\eqref{eq:set_omega} and the frequency period $K=100$ in~\eqref{eq:iw} was used (Note: the iterative reweighting of SOAV model was initially proposed \cite{hayakawa_convex_2017,hayakawa_reconstruction_2018} with an ADMM-type algorithm).
From Fig.~\ref{fig:ber_performances}~\subref{fig:iw}, `IW-cLiGME' improves `cLiGME', while even `cLiGME' outperforms `IW-SOAV'.

To examine the impact of choices of $(\beta_k)_{k=0}^{\infty}$ in generalized superiorization~\eqref{eq:perturbation}, we compared generalized superiorization of `cLiGME' with
(i)~$\beta_{k}=0$ (which reduces to the original `cLiGME'),
(ii)~$\beta_{k}=0.1\cdot 0.99^k$ ($(\beta_{k})_{k=0}^{\infty}$ is summable),
(iii)~$\beta_{k}=0.1\cdot k^{-1/2}$ ($(\beta_{k})_{k=0}^{\infty}$ is nonsummable but $\beta_{k}\to0$ $(k\to\infty)$), and
(iv)~$\beta_{k}=0.01$ (constant).
Fig.~\ref{fig:perturbations} shows history of BER achieved by generalized superiorization of `cLiGME' with such $(\beta_{k})_{k=0}^{\infty}$ in~\eqref{eq:perturbation}, where $\text{SNR}=20~\text{dB}$, $\mu=10^{-4}$ and $\omega^{\ev{l}}_{n}=1/8$ $(l=1,2,\ldots,8;\,n=1,2,\ldots,N)$.
From Fig.~\ref{fig:perturbations}, $\beta_{k}=0.01$ outperforms the others.
Fig.~\ref{fig:ber_performances}~\subref{fig:sup} shows BER, at each SNR, of `cLiGME' and generalized superiorization of `cLiGME' (denoted by `GS-cLiGME') with $\beta_{k}=0.01$, where the averaged BER $0$ was replaced by the machine epsilon.
From Fig.~\ref{fig:ber_performances}~\subref{fig:sup}, we see `GS-cLiGME' improves `cLiGME'.
Although these two modifications do not have convergence guarantees, they seem to improve dramatically the estimation accuracy of Algorithm~\ref{alg}.

\section{Conclusion}
\label{sec:conclusion}

We proposed a novel contrastive regularizer that has designated isolated local minimizers for discrete-valued estimation problems.
By using the proposed regularizer, we also presented an iterative algorithm with guaranteed convergence to a global minimizer of the contrastively regularized least squares model under the overall convexity of the cost function.
Numerical experiments demonstrate that the proposed model and algorithm have great potential for the discrete-valued signal estimation problem, and that two simple techniques effectively improve the estimation performance.

\section*{Acknowledgement}
\label{sec:acknowledgement}
The authors would like to thank Prof. Kazunori Hayashi (Kyoto University) and Prof. Masahiro Yukawa (Keio University) for valuable discussions at a conference (APSIPA 2024).
This work was supported by JSPS Grants-in-Aid (19H04134,23KJ0945)
and by JST SICORP (JPMJSC20C6).

\bibliographystyle{ieicetr}%
\bibliography{main}%
\numberwithin{equation}{section}
\def\theequation{\@Alph\c@section$\cdot$\,\@arabic\c@equation}
\appendix
\label{appendix}

\section{Known facts}
\label{appendix:known_facts}

\noindent\textbf{(Subdifferential~\cite[Def. 16.1]{bauschke_convex_2017})}
For a function $f\in\Gamma_{0}(\mathcal{H})$, the subdifferential of $f$ is defined as the set valued operator
\begin{equation}
  \myskip{0.3}
  \hspace{-2em}
  \partial f:\mathcal{H}\to 2^{\mathcal{H}}:
  x\mapsto\{u\in\mathcal{H}\mid \ev{y-x,u}_{\mathcal{H}}+f(x)\leq f(y),\forall y\in\mathcal{H}\}.
\end{equation}
The subdifferential has the following properties:
\begin{fact}[Some properties of subdifferential]\label{fct:subdifferential}
  \,
  \begin{enumerate}[label=(\alph*), leftmargin=*]
    \item \label{fct:fermat_rule}
          (Fermat's rule~\cite[Thm. 16.3]{bauschke_convex_2017}).
          Let $f\in\Gamma_{0}(\mathcal{H})$ and $\widebar{x}\in\mathcal{H}$.
          Then
          \begin{equation}
            \label{eq:fermat_rule}
            \widebar{x}\in\argmin_{x\in\mathcal{H}}f(x)
            \Leftrightarrow
            0_{\mathcal{H}}\in\partial f(\widebar{x}).
          \end{equation}
    \item \label{fct:sum_rule}
          (Sum rule~\cite[Cor. 16.48 (iii)]{bauschke_convex_2017}).
          Let $f,g\in\Gamma_{0}(\mathcal{H})$ satisfy $\dom g=\mathcal{H}$.
          Then
          \begin{equation}
            \label{eq:sum_rule}
            \partial(f+g)=\partial f +\partial g.
          \end{equation}
    \item \label{fct:chain_rule}
          (Chain rule~\cite[Cor. 16.53]{bauschke_convex_2017}).
          Let $g\in\Gamma_0(\mathcal{K})$ and $\mathcal{L}\in\mathcal{B}(\mathcal{H},\mathcal{K})$ satisfy $0_{\mathcal{K}}\in\ri(\dom g-\ran(\mathcal{L}))$.
          Then
          \begin{equation}
            \label{eq:chain_rule}
            \partial(g\circ \mathcal{L})=\mathcal{L}^{\ast}\circ(\partial g)\circ \mathcal{L}.
          \end{equation}
    \item \label{fct:differential}
          (Subdifferential and G\^{a}teaux differential~\cite[Prop. 17.31 (i)]{bauschke_convex_2017}).
          Let $f\in\Gamma_{0}(\mathcal{H})$ and $x\in\dom f$.
          Suppose that $f$ is G\^{a}teaux differentiable%
          \footnote{
          (G\^{a}teaux differentiable)
          Let $U$ be an open subset of $\mathcal{H}$. Then a function $f:U\to\mathbb{R}$ is said to be G\^{a}teaux differentiable at $x\in U$ if there exists $a(x)\in\mathcal{H}$ such that
          $
            (\forall h\in\mathcal{H})\quad
            \lim_{t\to 0}\frac{f(x+t h)-f(x)}{t}=\ev{a(x),h}_{\mathcal{H}}.
          $
          In this case, $\nabla f(x)\coloneqq a(x)$ is called gradient of $f$ at $x$.
          } at $x$.
          Then
          \begin{equation}
            \label{eq:differential}
            \partial f(x)=\{\nabla f(x)\}.
          \end{equation}
    \item(Subdifferential and conjugate~\cite[Cor. 16.30]{bauschke_convex_2017}).
          Let $f\in\Gamma_{0}(\mathcal{H})$.
          Then, for any $(x,u)\in\mathcal{H}\times\mathcal{H}$,
          \begin{equation}
            \label{eq:conjugate_subdifferential}
            u\in\partial f(x)
            \Leftrightarrow
            x\in \partial f^{\ast}(u).
          \end{equation}
  \end{enumerate}
\end{fact}

\noindent\textbf{(Maximally monotone operator~\cite[Def. 20.20]{bauschke_convex_2017})}
A set-valued operator $A:\mathcal{H}\to 2^{\mathcal{H}}$ is said to be \emph{maximally monotone} if, for every $(x,u)\in\mathcal{H}\times\mathcal{H}$,
\begin{equation}
  \label{eq:maximally_monotone}
  \hspace{-2em}
  \myskip{0.9}
  (x,u)\in\gra(A)
  \Leftrightarrow
  (\forall(y,v)\in\gra(A))
  \ev{x-y,u-x}_{\mathcal{H}}\geq 0,
\end{equation}
where $\gra(A)\coloneqq\{(x,u)\in\mathcal{H}\times\mathcal{H}\mid u\in A(x)\}$.
\begin{fact}[Some properties of maximally monotone operators]
  \label{fct:maximally_monotone}
  \,
  \begin{enumerate}[label=(\alph*), leftmargin=*]
    \item \label{item:MM_subdifferential}
          (\cite[Thm. 20.48]{bauschke_convex_2017}).
          The subdifferential $\partial f:\mathcal{H}\to 2^{\mathcal{H}}$ of $f\in\Gamma_{0}(\mathcal{H})$ is maximally monotone.
    \item \label{item:MM_skew}
          (\cite[Exm. 20.35]{bauschke_convex_2017}).
          If $\mathcal{L}\in\mathcal{B}(\mathcal{H},\mathcal{H})$ is skew-symmetry (i.e., $\mathcal{L}^{\ast}=-\mathcal{L}$),
          then $\mathcal{L}$ is maximally monotone.
    \item \label{item:MM_sum}
          (\cite[Cor. 25.5 (i)]{bauschke_convex_2017}).
          For maximally monotone operators $A,B\in\mathcal{H}\to 2^{\mathcal{H}}$, $A+B\in\mathcal{H}\to 2^{\mathcal{H}}$ is maximally monotone if $\dom(B)\coloneqq\{x\in\mathcal{H}\mid B(x)\neq\emptyset\}=\mathcal{H}$.
    \item \label{item:MM_hilbert}
          (\cite[Prop. 20.24]{bauschke_convex_2017}).
          If $A\in\mathcal{H}\to 2^{\mathcal{H}}$ is maximally monotone over $\Hilbert{\mathcal{H}}{\mathcal{H}}$ and $\mathcal{L}\succ O_{\mathcal{H}}$, then $\mathcal{L}^{-1}\circ A$ is maximally monotone over $\Hilbert{\mathcal{H}}{\mathcal{L}}$.
  \end{enumerate}
\end{fact}

\section{Proof of Lemma~\ref{lemma:separable}}
\label{appendix:proof_separable}

\ref{item:separable:local}
  $(\Rightarrow)$
  Suppose that
  $\x^{\heartsuit} \in f$
  is a local minimizer of
  $f$,
  i.e., there exists
  $\delta \in \mathbb{R}_{++}$
  such that
  $f(\x^{\heartsuit}) \leq f(\x)$
  for all
  $\x \in \mathcal{N}_{\x^{\heartsuit}}\coloneqq \mybigtimes_{n=1}^{N}B(x_{n}^{\heartsuit}, \delta) \coloneqq \mybigtimes_{n=1}^{N}\{x \in \mathbb{R} \mid \abs{x-x_ {n}^{\heartsuit}} < \delta\}$.
  For contradiction, we assume that
  $x^{\heartsuit}_{\widebar{n}}$
  is not a local minimizer of
  $f_{\widebar{n}}$
  for some
  $\widebar{n} \in \{1,2,\ldots,N\}$.
  Without loss of generality, assume
  $\widebar{n} = 1$
  for simplicity.
  In this case, there exists
  $\widebar{x}_{1} \in B(x_{1}^{\heartsuit}, \delta)$
  such that
  $f_{1}(x^{\heartsuit}_{1})> f_{1}(\widebar{x}_{1})$.
  For
  $\widebar{\x} \coloneqq (\widebar{x}_{1}, x^{\heartsuit}_{2}, x^{\heartsuit}_{3},\ldots,x^{\heartsuit}_{N}) \in \mathcal{N}_{\x^{\heartsuit}}$,
  we have
  $f(\x^{\heartsuit}) = \sum_{n=1}^{N} f_{n}(x^{\heartsuit}_{n}) > f_{1}(\widebar{x}_{1}) + \sum_{n=2}^{N} f_{n}(x^{\heartsuit}_{n}) = f(\widebar{\x})$,
  which is absurd.

  $(\Leftarrow)$
  For every
  $n=1,2,\ldots,N$,
  suppose that
  $x^{\heartsuit}_{n} \in \mathbb{R}$
  is a local minimizer of
  $f_{n}$.
  Then, there exists
  $\delta > 0$
  such that, for all
  $n = 1,2,\ldots,N$,
  $f_{n}(x^{\heartsuit}_{n}) \leq f(x)\ (\forall x \in B(x_{n}^{\heartsuit}, \delta))$.
  By letting
  $\mathcal{N}_{\x^{\heartsuit}} \coloneqq \mybigtimes_{n=1}^{N}B(x_{n}^{\heartsuit}, \delta)$,
  we have
  $f(\x^{\heartsuit}) = \sum_{n=1}^{N} f_{n}(x^{\heartsuit}_{n}) \leq \sum_{n=1}^{N} f_{n}(x_{n}) = f(\x)$
  for all
  $\x \in \mathcal{N}_{\x^{\heartsuit}}$,
  implying thus
  $\x^{\heartsuit}$
  is a local minimizer of
  $f$.

  \ref{item:separable:isolated}
  $(\Rightarrow)$
  Suppose that
  $\x^{\heartsuit} \in f$
  is an isolated local minimizer of
  $f$,
  i.e.,
  $\x^{\heartsuit}$
  is the only local minimizer of
  $f$
  in
  $\mathcal{N}_{\x^{\heartsuit}} \coloneqq \mybigtimes_{n=1}^{N}B(x_{n}^{\heartsuit}, \delta)$
  with some
  $\delta > 0$.
  For every
  $n = 1,2,\ldots,N$,
  choose arbitrarily a local minimizer
  $\widebar{x}_{n} \in B(x_{n}^{\heartsuit}, \delta)$
  of
  $f_{n}$.
  By~\ref{item:separable:local},
  $\widebar{\x} \coloneqq (\widebar{x}_{1},\widebar{x}_{2},\ldots,\widebar{x}_{n})^{\top}\in \mathcal{N}_{\x^{\heartsuit}}$
  is a local minimizer of
  $f$.
  Since
  $\x^{\heartsuit}$
  is the only local minimizer of
  $f$
  in
  $\mathcal{N}_{\x^{\heartsuit}}$,
  we have
  $\widebar{\x} = \x^{\heartsuit}$.
  Hence, for
  $n=1,2,\ldots,N$,
  $\widebar{x}_{n} = x^{\heartsuit}_{n}$
  holds true, and
  $x^{\heartsuit}_{n}$
  is an isolated local minimizer of
  $f_{n}$.

  $(\Leftarrow)$
  For every
  $n=1,2,\ldots,N$,
  suppose that
  $x^{\heartsuit}_{n}$
  is an isolated local minimizer of
  $f_{n}$.
  Then, there exists
  $\delta > 0$
  such that
  every
  $x^{\heartsuit}_{n}\ (n=1,2,\ldots,N)$
  is the only local minimizer of
  $f_{n}$
  in
  $B(x_{n}^{\heartsuit}, \delta)$.
  Choose a local minimizer
  $\widebar{\x} \in \mathcal{N}_{\x^{\heartsuit}} \coloneqq \mybigtimes_{n=1}^{N}B(x_{n}^{\heartsuit}, \delta)$
  of
  $f$
  arbitrarily.
  For every
  $n=1,2,\ldots,N$,
  since
  (i)
  $\widebar{x}_{n}\in B(x_{n}^{\heartsuit}, \delta)\ $
  is a local minimizer of
  $f_{n}$
  by~\ref{item:separable:local},
  and (ii)
  $x^{\heartsuit}_{n}$
  is the only local minimizer of
  $f_{n}$
  in
  $B(x_{n}^{\heartsuit}, \delta)$,
  we have
  $\widebar{x}_{n} = x^{\heartsuit}_{n}$.
  Hence,
  $\widebar{\x} = \x^{\heartsuit}$
  holds true, and
  $\x^{\heartsuit}$
  is an isolated local minimizer of
  $f$.

\section{Proof for Lemma~\ref{lem:relative_interior}}
\label{appendix:proof_lemma}

Since $\Psi(\in\Gamma_{0}(\mathcal{Z}))$ is coercive and $\frac{1}{2}\bnorm{B\cdot}_{\widetilde{\mathcal{Z}}}^{2}(\in\Gamma_{0}(\mathcal{Z}))$ is bounded below with $\dom\left(\frac{1}{2}\bnorm{B\cdot}_{\widetilde{\mathcal{Z}}}^{2}\right)=\mathcal{Z}$,
$\Psi+\frac{1}{2}\bnorm{B\cdot}_{\widetilde{\mathcal{Z}}}^{2}$ is coercive by \cite[Cor. 11.16 (ii)]{bauschke_convex_2017}.
Then \cite[Prop. 14.16]{bauschke_convex_2017} yields
$0_{\mathcal{Z}}\in\Int\left(\dom\left(\Psi+\frac{1}{2}\bnorm{B\cdot}_{\widetilde{\mathcal{Z}}}^{2}\right)^{\ast}\right)$.
Since
\begin{equation}
  \hspace{-2em}
  (0_{\mathcal{Z}}\in)
  \Int\left(\dom\left(\Psi+\frac{1}{2}\bnorm{B\cdot}_{\widetilde{\mathcal{Z}}}^{2}\right)^{\ast}\right)
  \subset
  \ri\left(\dom\left(\Psi+\frac{1}{2}\bnorm{B\cdot}_{\widetilde{\mathcal{Z}}}^{2}\right)^{\ast}\right)
\end{equation}
(see, e.g., \cite[(6.11)]{bauschke_convex_2017}) and $\ri(\ran(B^{\ast}))=\ran(B^{\ast})\ni 0_{\mathcal{Z}}$, we have
\begin{equation}
  0_{\mathcal{Z}}\in
  \left(
  \ri\left(\dom\left(\Psi+\frac{1}{2}\bnorm{B\cdot}_{\widetilde{\mathcal{Z}}}^{2}\right)^{\ast}\right)
  \right)\cap
  (\ri(\ran(B^{\ast}))),
\end{equation}
which further yields by \cite[Prop. 6.19~(viii)]{bauschke_convex_2017}
\begin{equation}
  0_{\mathcal{Z}}\in\ri\left(\dom\left(\left(\Psi+\frac{1}{2}\bnorm{B\cdot}_{\widetilde{\mathcal{Z}}}^{2}\right)^{\ast}\right)-
  \ran\left({B}^{\ast}\right)\right).
\end{equation}

\section{Proof of Proposition~\ref{prop:cLiGME_algorithm}}
\label{appendix:proof_theorem}

To follow the proof of [23, Thm. 1], we show the following claims.

\begin{claim}[Slight extension of {\cite[Claim D.1]{abe_linearly_2020}}]
  \label{claim:inclusion}
  For any $\widebar{x}\in\mathcal{X}$, we have $\widebar{x}\in\mathcal{S}_{C}\coloneqq\argmin_{x\in C}J_{\Psi_{B}\circ\mathfrak{L}}(x)$ if and only if there exists $(\widebar{v}, \widebar{w})\in\mathcal{Z}\times\mathcal{Z}$ such that
  \begin{equation}
    \label{eq:inclusion}
    (0_{\mathcal{X}}, 0_{\mathcal{Z}}, 0_{\mathcal{Z}})\in
    (F+G_{C})(\widebar{x}, \widebar{v}, \widebar{w}),
  \end{equation}
  where $F:\mathcal{H}\to\mathcal{H}$ and $G_{C}:\mathcal{H}\to 2^{\mathcal{H}}$ are defined by
  \begin{equation}\label{eq:FG}
    \myskip{0.1}
    \hspace{-2em}
    \begin{aligned}
      F(x, v, w)
      \coloneqq &
      (
      (A^{\ast}A - \mu\mathfrak{L}^{\ast}B^{\ast}B\mathfrak{L})x - A^{\ast}y,
      \mu B^{\ast}Bv,
      0_{\mathcal{Z}}
      ),
      \\
      G_{C}(x, v, w)
      \coloneqq &
      (\mu\mathfrak{L}^{\ast}B^{\ast}Bv + \mu\mathfrak{L}^{\ast}w + \partial\iota_{C}(x))
      \\ & \times
      (-\mu B^{\ast}B\mathfrak{L}x + \mu\partial\Psi(v))
      \\ & \times
      (-\mu\mathfrak{L}x + \mu\partial\Psi^{\ast}(w))
    \end{aligned}
  \end{equation}
  (Note: $F$ in~\eqref{eq:FG} is the same as $F$ in~\cite[(D.1)]{abe_linearly_2020},
  and $G_{C}$ in a special case $C\coloneqq\mathcal{X}$ reproduces $G$ employed in~\cite[(D.1)]{abe_linearly_2020}).
  \begin{proof}[{Proof of Claim~\ref{claim:inclusion}}]
    We present our proof by referring to the proof of~\cite[Claim D.1]{abe_linearly_2020} that assumes
    (i)~the even symmetry of $\Psi$, and
    (ii)~$C=\mathcal{X}$.
    The solution set $\mathcal{S}_{C}$ of Problem~\ref{prob:cLiGME_model} can be expressed as
    \begin{align}
       & \mathcal{S}_{C}
      =\argmin_{x\in \mathcal{X}}J_{\Psi_{B}\circ\mathfrak{L}}(x)+\iota_{C}(x) \\
       & \overset{\eqref{eq:fermat_rule}}{=}
      \{ x\in \mathcal{X} \mid 0_{\mathcal{X}} \in \partial(J_{\Psi_{B}\circ\mathfrak{L}} + \iota_{C})(x)\}
      \\
       & \overset{\eqref{eq:sum_rule}}{=}
      \{ x \in \mathcal{X} \mid 0_{\mathcal{X}} \in \partial J_{\Psi_{B}\circ\mathfrak{L}}(x) + \partial \iota_{C}(x)\},
      \label{eq:zero_inclusion}
    \end{align}
    where $\iota_{C}$ is the indicator function of $C$ (see~\eqref{eq:indicator}), and
    $\dom J_{\Psi_{B}\circ\mathfrak{L}} = \mathcal{X}$ is used in the last equality.

    By the discussion in~\cite{abe_linearly_2020} to derive~\cite[(D.9)]{abe_linearly_2020}, we get%
    \footnote{
      In~\cite{abe_linearly_2020}, the expression in~\eqref{eq:partial_J} is derived under the even symmetry of $\Psi$, where the even symmetry is used for applying the chain rule~\eqref{eq:chain_rule} in Fact~\ref{fct:subdifferential}~\ref{fct:chain_rule} to $\partial\left[\left(\Psi + \frac{1}{2}\bnorm{B\cdot}_{\widetilde{\mathcal{Z}}}^{2}\right)^{\ast}\circ B^{\ast}\right]$.
      However, we can apply the chain rule to $\partial\left[\left(\Psi + \frac{1}{2}\bnorm{B\cdot}_{\widetilde{\mathcal{Z}}}^{2}\right)^{\ast}\circ B^{\ast}\right]$ without assuming the even symmetry because the constraint qualification~\eqref{eq:ri} in Lemma~\ref{lem:relative_interior} for the chain rule is always achieved by assuming only the coercivity of $\Psi$ and $\dom(\Psi) = \mathcal{Z}$.
    }
    \begin{equation}\label{eq:partial_J}
      \myskip{0.0}
      \hspace{-2em}
      \begin{aligned}
        \partial J_{\Psi_{B}\circ\mathfrak{L}}(x) &=
        (A^{\ast}A - \mu\mathfrak{L}^{\ast}B^{\ast}B\mathfrak{L})x
        - A^{\ast}y+ \mu\mathfrak{L}^{\ast}\partial\Psi(\mathfrak{L}x)
        \\
         & +
        \mu(B^{\ast}B\mathfrak{L})^{\ast}
        \partial\left(\Psi + \frac{1}{2}\bnorm{B\cdot}_{\widetilde{\mathcal{Z}}}^{2}\right)^{\ast}
        (B^{\ast}B\mathfrak{L}x).
      \end{aligned}
    \end{equation}

    By~\eqref{eq:conjugate_subdifferential}, we have for $\widebar{v},\widebar{w}\in\mathcal{Z}$,
    \begin{equation}
      \label{eq:vw}
      \myskip{0.2}
      \hspace{-2em}
      \begin{cases}
        \widebar{v}\in\partial\left(\Psi + \frac{1}{2}\bnorm{B\cdot}_{\widetilde{\mathcal{Z}}}^{2}\right)^{\ast}(B^{\ast}B\mathfrak{L}\widebar{x})
        \Leftrightarrow
        B^{\ast}B\mathfrak{L}\widebar{x}\in
        \smash{
        \underbrace{
          \partial\left(\Psi + \frac{1}{2}\bnorm{B\cdot}_{\widetilde{\mathcal{Z}}}^{2}\right)(\widebar{v})
        }_{\overset{(\spadesuit)}{=}\partial\Psi(\widebar{v}) + B^{\ast}B\widebar{v}}
        }
        \\
        \\
        \widebar{w}\in\partial\Psi(\mathfrak{L}\widebar{x})\Leftrightarrow\mathfrak{L}\widebar{x}\in\partial\Psi^{\ast}(\widebar{w}),
      \end{cases}
    \end{equation}
    where $(\spadesuit)$ follows from~\eqref{eq:sum_rule} with $\dom\left(\norm{B\cdot}_{\widetilde
        {\mathcal{Z}}}\right)=\mathcal{Z}$ and~\eqref{eq:differential} with $\nabla (\frac{1}{2}\bnorm{B\cdot}_{\widetilde{\mathcal{Z}}}^{2})=B^{\ast}B(\cdot)$.
    From~\eqref{eq:zero_inclusion} and~\eqref{eq:partial_J}, we have
      {
        \myskip{0.1}
        \begin{align}
          \hspace{-2em}
          \widebar{x}\in\mathcal{S}_{C}
           &
          \overset{\hphantom{\eqref{eq:vw}}}{\Leftrightarrow} 0_{\mathcal{X}} \in \partial J_{\Psi_{B}\circ\mathfrak{L}}(\widebar{x}) + \partial \iota_{C}(\widebar{x})
          \\
          \hspace{-2em}
           & \overset{\eqref{eq:vw}}{\Leftrightarrow}
          (\exists \widebar{v},\widebar{w}\in\mathcal{Z})\
          \begin{cases}
            0_{\mathcal{X}}\in
            (A^{\ast}A-\mu\mathfrak{L}^{\ast}B^{\ast}B\mathfrak{L})\widebar{x}
            -A^{\ast}y
            \\ \hfill
            + \mu\mathfrak{L}^{\ast}B^{\ast}B\widebar{v} + \mu\mathfrak{L}^{\ast}\widebar{w} + \partial\iota_{C}(\widebar{x}),
            \\
            B^{\ast}B\mathfrak{L}\widebar{x}\in
            \partial\Psi(\widebar{v}) + B^{\ast}B\widebar{v},
            \\
            \mathfrak{L}\widebar{x}\in
            \partial\Psi^{\ast}(\widebar{w})
          \end{cases}
          \\
          \hspace{-2em}
           &
          \overset{\hphantom{\eqref{eq:vw}}}{\Leftrightarrow}
          (\exists \widebar{v},\widebar{w}\in\mathcal{Z})\
          \begin{cases}
            0_{\mathcal{X}}\in
            (A^{\ast}A-\mu\mathfrak{L}^{\ast}B^{\ast}B\mathfrak{L})\widebar{x}
            -A^{\ast}y
            \\ \hfill
            + \mu\mathfrak{L}^{\ast}B^{\ast}B\widebar{v} + \mu\mathfrak{L}^{\ast}\widebar{w} + \partial\iota_{C}(\widebar{x}),
            \\
            0_{\mathcal{Z}}\in
            -\mu B^{\ast}B\mathfrak{L}\widebar{x} + \mu B^{\ast}B\widebar{v} + \mu\partial\Psi(\widebar{v}),
            \\
            0_{\mathcal{Z}}\in
            -\mu\mathfrak{L}\widebar{x} + \mu\partial\Psi^{\ast}(\widebar{w})
          \end{cases}
          \\
          \hspace{-2em}
           &
          \overset{\hphantom{\eqref{eq:vw}}}{\Leftrightarrow}
          (\exists \widebar{v},\widebar{w}\in\mathcal{Z})\
          (0_{\mathcal{X}}, 0_{\mathcal{Z}}, 0_{\mathcal{Z}})\in
          (F+G_{C})(\widebar{x}, \widebar{v}, \widebar{w}).
          \hspace{1em}\qedhere
        \end{align}
      }%
  \end{proof}
\end{claim}

\begin{claim}
  \label{claim:maximally_monotone}
  $\mathfrak{P}^{-1}\circ G_{C}$ is maximally monotone over the Hilbert space $\Hilbert{\mathcal{H}}{\mathfrak{P}}$, where $\mathfrak{P}\succ O_{\mathcal{H}}$ has been shown in~\cite[(34) in Thm. 1 (b)]{abe_linearly_2020}.
  \begin{proof}[{Proof of Claim~\ref{claim:maximally_monotone}}]
    Consider the decomposition of $G_{C}= G_{C}^{\ev{1}} + G_{C}^{\ev{2}}$ in~\eqref{eq:FG}, where $G_{C}^{\ev{1}}:\mathcal{H}\to 2^{\mathcal{H}}:(x,v,w)\mapsto(\partial\iota_{C}(x))\times(\mu\partial\Psi(v))\times(\mu\partial\Psi^{\ast}(w))$
    and $G_{C}^{\ev{2}}:\mathcal{H}\to\mathcal{H}:(x,v,w)\mapsto(\mu\mathfrak{L}^{\ast}B^{\ast}Bv + \mu\mathfrak{L}^{\ast}w, -\mu B^{\ast}B\mathfrak{L}x, -\mu\mathfrak{L}x)$ are maximally monotone%
    \footnote{
    $G_{C}^{\ev{1}}$ is maximally monotone over
    $\Hilbert{\mathcal{H}}{\mathcal{H}}$ because
    (i)~$G_{C}^{\ev{1}}$ is also given by the subdifferential of $g\in \Gamma_{0}(\mathcal{H})$ defined as $g:\mathcal{H}\to(-\infty,\infty]:(x,v,w)\mapsto \iota_{C}(x) + \mu\Psi(v) + \mu\Psi^{\ast}(w)$~\cite[Preposition 16.9]{bauschke_convex_2017}; and
    (ii)~$\partial g:\mathcal{H}\to2^{\mathcal{H}}$ is maximally monotone by Fact~\ref{fct:maximally_monotone}~\ref{item:MM_subdifferential}.
    $G_{C}^{\ev{2}}$ is a bounded linear skew-symmetric operator, i.e., ${G_{C}^{\ev{2}}}^{\ast}=-G_{C}^{\ev{2}}$, and thus $G_{C}^{\ev{2}}$ is maximally monotone over
    $\Hilbert{\mathcal{H}}{\mathcal{H}}$
    by Fact~\ref{fct:maximally_monotone}~\ref{item:MM_skew}.
    }
    over $\Hilbert{\mathcal{H}}{\mathcal{H}}$.
    By Fact~\ref{fct:maximally_monotone}~\ref{item:MM_sum} and $\dom\left(G_{C}^{\ev{2}}\right)=\mathcal{H}$, $G_{C}=G_{C}^{\ev{1}}+G_{C}^{\ev{2}}$ is maximally monotone over
    $\Hilbert{\mathcal{H}}{\mathcal{H}}$.
    Finally, by Fact~\ref{fct:maximally_monotone}~\ref{item:MM_hilbert}, $\mathfrak{P}^{-1}\circ G_{C}$ is maximally monotone over  $\Hilbert{\mathcal{H}}{\mathfrak{P}}$.
  \end{proof}
\end{claim}

\begin{proof}[Proof of Proposition~\ref{prop:cLiGME_algorithm}]
  \ref{item:fixed_point}
  From~\eqref{eq:TcLiGME}, we have
    {
      \myskip{0.0}
      \begin{align}
        \hspace{-2em}
         &
        \label{eq:xi_zeta_eta}
        T_{\cLiGME}(x,v,w) = (\xi,\zeta,\eta)
        \\
        \label{eq:inclusion2}
        \hspace{-2em}
        \Leftrightarrow
         &
        \begin{cases}
          (\sigma\Id-(A^{\ast}A-\mu\mathfrak{L}^{\ast}B^{\ast}B\mathfrak{L}))x
          \\ \hfill
          -\mu\mathfrak{L}^{\ast}B^{\ast}Bv
          -\mu\mathfrak{L}^{\ast}w
          +A^{\ast}y
          \in[\sigma\Id + \partial\iota_{C}](\xi),
          \\
          \mu B^{\ast}B\mathfrak{L}(2\xi-x)
          +(\tau\Id-\mu B^{\ast}B)v
          \in[\tau\Id + \mu\partial\Psi](\zeta),
          \\
          \mu\mathfrak{L}(2\xi-x)+\mu w
          \in[\mu\Id + \mu\partial\Psi^{\ast}](\eta)
        \end{cases} \\
        \hspace{-2em}
        \Leftrightarrow
         &
        \label{eq:P_F_P_G}
        (\mathfrak{P}-F)(x,v,w)\in
        (\mathfrak{P}+G_{C})(\xi,\zeta,\eta),
      \end{align}}%
  where~\eqref{eq:inclusion2} follows from, for $f \in \Gamma_{0}(\mathcal{H})$, $\Prox_{f} = (\Id+\partial f)^{-1}$ \cite[Prop. 16.44]{bauschke_convex_2017}
  and the expression $P_{C}=\Prox_{\frac{1}{\sigma}\iota_{C}}$ (see Example~\ref{ex:proximity_operator}~\ref{item:indicator}).
  By Claim~\ref{claim:inclusion}, we have
  \begin{align}
     & \widebar{x}\in \mathcal{S}_{C}                                                                                  \\
    \Leftrightarrow
     & (\exists \widebar{v},\widebar{w}\in\mathcal{Z})\
    0_{\mathcal{X}}\in (F+G_{C})(\widebar{x}, \widebar{v}, \widebar{w})                                                \\
    \Leftrightarrow
     & (\exists \widebar{v},\widebar{w}\in\mathcal{Z})\
    (\mathfrak{P}-F)(\widebar{x},\widebar{v},\widebar{w})\in (\mathfrak{P}+G_{C})(\widebar{x},\widebar{v},\widebar{w}) \\
    \Leftrightarrow
     & \label{eq:fixed}
    (\exists \widebar{v},\widebar{w}\in\mathcal{Z})\
    T_{\cLiGME}(\widebar{x},\widebar{v},\widebar{w})=(\widebar{x},\widebar{v},\widebar{w})                             \\
    \Leftrightarrow
     & (\exists \widebar{v},\widebar{w}\in\mathcal{Z})\
    (\widebar{x},\widebar{v},\widebar{w})\in\Fix(T_{\cLiGME})                                                          \\
    \Leftrightarrow
     & \widebar{x}\in\Xi(\Fix(T_{\cLiGME})),
  \end{align}
  where we used the relation \eqref{eq:xi_zeta_eta}$\Leftrightarrow$\eqref{eq:P_F_P_G} in \eqref{eq:fixed}.

  \ref{item:nonexpansive}
  By~\eqref{eq:xi_zeta_eta}$\Leftrightarrow$\eqref{eq:P_F_P_G} and $\mathfrak{P}\succ O_{\mathcal{H}}$,
  {\myskip{0.5}
      \begin{align}
        \hspace{-2em} & T_{\cLiGME}(x,v,w) = (\xi,\zeta,\eta)
        \\
        \hspace{-2em}\Leftrightarrow
                      & \label{eq:p_inverse}
        (\Id- \mathfrak{P}^{-1}\circ F)(x,v,w)\in
        (\Id+ \mathfrak{P}^{-1}\circ G_{C})(\xi,\zeta,\eta).
      \end{align}}%
  By Claim~\ref{claim:maximally_monotone} and~\cite[Cor. 23.11 (i)]{bauschke_convex_2017}, $(\Id+ \mathfrak{P}^{-1}\circ G_{C})^{-1}$ is a  $\frac{1}{2}$-averaged nonexpansive (single-valued) mapping over
  $\Hilbert{\mathcal{H}}{\mathfrak{P}}$, and therefore we obtain from~\eqref{eq:p_inverse}
  \begin{equation}
    T_{\cLiGME}=
    (\Id+ \mathfrak{P}^{-1}\circ G_{C})^{-1}\circ(\Id- \mathfrak{P}^{-1}\circ F),
  \end{equation}
  where $\frac{1}{\kappa}$-averaged nonexpansiveness of $(\Id- \mathfrak{P}^{-1}\circ F)$ over $\Hilbert{\mathcal{H}}{\mathfrak{P}}$ has been shown in the proof of~\cite[Thm. 1 (b)]{abe_linearly_2020}.
  By~\cite[Thm. 3 (b)]{ogura_non-strictly_2002}\cite[Prop. 2.4]{combettes_compositions_2015}, $T_{\cLiGME}$ is $\frac{\kappa}{2\kappa-1}$-averaged nonexpansive in
  $\Hilbert{\mathcal{H}}{\mathfrak{P}}$.

  \ref{item:convergence}
  Thanks to \ref{item:fixed_point} and \ref{item:nonexpansive}, the direct application of Krasnosel'ski\u{\i}-Mann iteration in Fact~\ref{fct:picard} to $T_{\cLiGME}$ yields~\ref{item:convergence}.
\end{proof}

\section{Proof of Theorem~\ref{thm:isolated}}
\label{appendix:proof_isolated}

  \ref{item:expression}
  By using the expression%
  \footnote{
    This expression can be verified with
      {
        \myskip{0.5}
        \begin{align}
           & \bigl(\norm{\cdot}_{\bm{\omega}^{\ev{l}};1}\bigr)_{\B^{\ev{l}}}(\x)                                                                                                                 \\
           & \overset{\B^{\ev{l}}=b\I_{N}}{=} \norm{\x}_{\bm{\omega}^{\ev{l}};1} - \min_{\v \in \mathbb{R}^{N}}\left[\norm{\v}_{\bm{\omega}^{\ev{l}};1} + \frac{b^{2}}{2}\norm{\x-\v}^{2}\right] \\
           & = \sum_{n=1}^{N}\omega^{\ev{l}}_{n}\left[\abs{x_{n}} - \min_{\v \in \mathbb{R}^{N}}\left(\abs{v_{n}} + \frac{b^{2}}{2\omega^{\ev{l}}_{n}}(x_{n}-v_{n})^{2}\right)\right]            \\
           & = \sum_{n=1}^{N}\omega^{\ev{l}}_{n}\left((\abs{\cdot})_{\frac{b}{\sqrt{\omega^{\ev{l}}_{n}}}}\right)(x_{n})
          = \sum_{n=1}^{N}\omega^{\ev{l}}_{n}\MC{\frac{b^{2}}{\omega^{\ev{l}}_{n}}}(x_{n}),
        \end{align}}%
    where we used the relation $(\abs{\cdot})_{\sqrt{\gamma}}=\MC{\gamma}$ in the last equality.
  }
  \begin{align}
     & (\x=(x_{1},x_{2},\ldots,x_{N})^{\top}\in\mathbb{R}^{N})                                               \\
     & \bigl(\norm{\cdot}_{\bm{\omega}^{\ev{l}};1}\bigr)_{\B^{\ev{l}}}(\x)=\sum_{n=1}^{N}\omega^{\ev{l}}_{n}
    \MC{\frac{b^{2}}{\omega^{\ev{l}}_{n}}}(x_{n}),
  \end{align}
  we have the following expression
  \begin{align}
     & (\x \in \mathbb{R}^{N}) \quad
    \Theta_{\LiGME}(\x)
    = \sum_{l=1}^{L}\bigl(\norm{\cdot}_{\bm{\omega}^{\ev{l}};1}\bigr)_{\B^{\ev{l}}}\bigl(\x-a^{\ev{l}}\ones_{N}\bigr) \\
     & = \sum_{l=1}^{L}\sum_{n=1}^{N}\omega^{\ev{l}}_{n}\MC{\frac{b^{2}}{\omega^{\ev{l}}_{n}}}(x_{n}-a^{\ev{l}})
    = \sum_{n=1}^{N} \psi_{n}(x_{n}),
  \end{align}

  \ref{item:designated}
  By~\ref{item:expression} and Lemma~\ref{lemma:separable}~\ref{item:separable:isolated},
  $\x \in \mathbb{R}^{N}$
  is an isolated local minimizer of
  $\Theta_{\LiGME}$
  if and only if every
  $x_{n} \in \mathbb{R}\ (n=1,2,\ldots,N)$
  is an isolated local minimizer of
  $\psi_{n}$.
  In the following, by fixing
  $n = 1,2,\ldots,N$,
  we show that
  $\mathfrak{A}$
  is the set of all isolated local minimizers of
  $\psi_{n}$.
  For
  $l = 1,2,\ldots,L$,
  let
  $\mathcal{N}_{a^{\ev{l}},n}\coloneqq \{x \in \mathbb{R} \mid \abs{x-a^{\ev{l}}} < \frac{\omega^{\ev{l}}_{n}}{b^{2}}\}$
  be an open neighborhood of
  $a^{\ev{l}}$.
  From~\eqref{eq:MCP}, we obtain
  \begin{equation}
    \hspace{-2em}
    \myskip{0.0}
    (l=1,2,\ldots,L,\ x \notin \mathcal{N}_{a^{\ev{l}},n})\quad
    \MC{\frac{b^{2}}{\omega^{\ev{l}}_{n}}}(x-a^{\ev{l}})
    = \frac{\omega^{\ev{l}}_{n}}{2b^{2}}. \label{eq:constant_outside_neighborhood}
  \end{equation}
  Moreover, we have
  \begin{equation}
    \hspace{-2em}
    \myskip{0.2}
    (a^{\ev{l}},a^{\ev{l'}}\in\mathfrak{A} \ {\rm s.t.}\ a^{\ev{l}}\neq a^{\ev{l'}}) \quad
    \mathcal{N}_{a^{\ev{l}},n} \cap \mathcal{N}_{a^{\ev{l'}},n} = \emptyset \label{eq:distinct_neighborhood}
  \end{equation}
  by the following inequality
    {
      \myskip{0.2}
      \begin{align}
        \hspace{-2em}
        (x \in \mathcal{N}_{a^{\ev{l}},n}) \quad
        \abs{x-a^{\ev{l'}}}
         & \geq \abs{a^{\ev{l'}} - a^{\ev{l}}} - \abs{x-a^{\ev{l}}}
        > d_{\min} - \frac{\omega^{\ev{l}}_{n}}{b^{2}}              \\
        \hspace{-2em}
         & \geq d_{\min} - \frac{d_{\min}}{2}
        = \frac{d_{\min}}{2}
        \geq \frac{\omega^{\ev{l'}}_{n}}{b^{2}},
      \end{align}}%
  where we used
  $\frac{d_{\min}}{2} \geq \frac{\omega_{\max}}{b^{2}} \geq \frac{\omega^{\ev{k}}_{n}}{b^{2}}\ (k=l,l')$.

  \textbf{(Case 1: $x \in \mathcal{N}_{a^{\ev{l}},n}\ (l=1,2,\ldots,L)$)}
  Choose
  $a^{\ev{l}} \in \mathfrak{A}$
  arbitrarily, and let
  $x \in \mathcal{N}_{a^{\ev{l}},n}$.
  By~\eqref{eq:distinct_neighborhood}, we have
  $x \notin \mathcal{N}_{a^{\ev{l'}},n}$
  for all
  $a^{\ev{l'}} \in\mathfrak{A} \setminus \{a^{\ev{l}}\}$.
  From~\eqref{eq:constant_outside_neighborhood}, we obtain the expression
  \begin{equation}
    \hspace{-2em}
    \myskip{0.2}
    (x \in \mathcal{N}_{a^{\ev{l}},n}) \
    \psi_{n}(x) =
    \MC{\frac{b^{2}}{\omega^{\ev{l}}_{n}}}(x-a^{\ev{l}}) +
    \underbrace{\sum_{l'=1,l'\neq l}^{L}\frac{\omega^{\ev{l'}}_{n}}{2b^{2}}}_{\text{const.}},
  \end{equation}
  implying thus
  $\psi_{n}$
  is differentiable on
  $\mathcal{N}_{a^{\ev{l}},n} \setminus \{a^{\ev{l}}\}$
  with the derivative
  $\psi'_{n}(x) = \mathrm{sgn}(x-a^{\ev{l}}) - \frac{b^{2}}{\omega^{\ev{l}}_{n}}(x-a^{\ev{l}})$
  for
  $x \in \mathcal{N}_{a^{\ev{l}},n} \setminus \{a^{\ev{l}}\}$.
  For
  $x \in \mathcal{N}_{a^{\ev{l}},n} \setminus \{a^{\ev{l}}\}$,
  due to
  $\babs{\frac{b^{2}}{\omega^{\ev{l}}_{n}}(x-a^{\ev{l}})} < 1$,
  we have
  $\psi'_{n}(x) \neq 0$,
  implying thus any
  $x \in \mathcal{N}_{a^{\ev{l}},n} \setminus \{a^{\ev{l}}\}$
  can not be a local minimizer of
  $\psi_{n}$.
  On the other hand,
  $a^{\ev{l}}$ is the unique minimizer of $x\mapsto \MC{\frac{b^{2}}{\omega^{\ev{l}}_{n}}}(x-a^{\ev{l}})$,
  and thus
  $a^{\ev{l}}$ is the unique isolated local minimizer of
  $\psi_{n}$
  in
  $\mathcal{N}_{a^{\ev{l}},n}$.

  \textbf{(Case 2: $x \in \mathbb{R}\setminus (\bigcup_{l=1}^{L}\mathcal{N}_{a^{\ev{l'}}})$)}
  By~\eqref{eq:constant_outside_neighborhood},
  we have
  \begin{equation}
    \left(x \in \mathbb{R}\setminus \left(\bigcup_{l=1}^{L}\mathcal{N}_{a^{\ev{l}},n}\right)\right) \quad
    \psi_{n}(x) =
    \underbrace{\sum_{l=1}^{L}\frac{\omega^{\ev{l}}_{n}}{2b^{2}}}_{\text{const.}},
  \end{equation}
  from which any
  $x \in \mathbb{R}\setminus \left(\bigcup_{l=1}^{L}\mathcal{N}_{a^{\ev{l}},n}\right)$
  can not be an isolated local minimizer of
  $\psi_{n}$.

  By combining Case 1 and Case 2, the set of all isolated local minimizers of
  $\psi_{n}$
  is
  $\mathfrak{A}$.

\section{Proof of Lemma~\ref{lem:LiGME_regularizer}}
\label{appendix:proof_LiGME_regularizer}

  \ref{item:commutativity}
  The equality in~\eqref{eq:GME_shift} follows from
  {
    \myskip{0.2}
    \begin{align}
      \hspace{-1em}
       & [\Phi(\cdot-z)]_{B}(u)                    \\
      \hspace{-1em}
      \overset{
        \hphantom{v'\coloneqq v-z}
      }{=}
       & \Phi(u-z) - \min_{v\in\mathcal{Z}}\left[
      \Phi(v-z)+\frac{1}{2}\bnorm{B(u-v)}_{\widetilde{\mathcal{Z}}}^{2}
      \right]                                      \\
      \hspace{-1em}
      \overset{
        \hphantom{v'\coloneqq v-z}
      }{=}
       & \Phi(u-z) - \min_{v\in\mathcal{Z}}\left[
      \Phi(v-z)+\frac{1}{2}\bnorm{B((u-z)-(v-z))}_{\widetilde{\mathcal{Z}}}^{2}
      \right]                                      \\
      \hspace{-1em}
      \overset{v'\coloneqq v-z}{=}
       & \Phi(u-z) - \min_{v'\in\mathcal{Z}}\left[
      \Phi(v')+\frac{1}{2}\bnorm{B((u-z)-v')}_{\widetilde{\mathcal{Z}}}^{2}
      \right]                                      \\
      \hspace{-1em}
      \overset{
        \hphantom{v'\coloneqq v-z}
      }{=}
       & \Phi_{B}(u-z).
    \end{align}
  }%

  \ref{item:LiGME_regularizer}
  By using a reformulation in~\cite[Exm.~3]{abe_linearly_2020}, we have
    \begin{align}
      (\x\in\mathbb{R}^{N})\quad
       &
      \Theta_{\LiGME}(\x)
      \\
       &
      \overset{
        \hphantom{\text{\ref{item:commutativity}}}
      }{=}
      \sum_{l=1}^{L}\bigl(\norm{\cdot}_{\bm{\omega}^{\ev{l}};1}\bigr)_{\B^{\ev{l}}}\bigl(\x-a^{\ev{l}}\ones_{N}\bigr)
      \\
       &
      \overset{\text{\ref{item:commutativity}}}{=}
      \sum_{l=1}^{L}\left[
      \norm{\cdot-a^{\ev{l}}\ones_{N}}_{\bm{\omega}^{\ev{l}};1}
      \right]_{\B^{\ev{l}}}\circ\Id(\x)
      \\
       &
      \overset{\text{\cite[Exm. 3]{abe_linearly_2020}}}{=}
      \Psi_{\B}\circ\mathfrak{L}(\x).
      \hspace{5em}\qedhere
    \end{align}

\section{Two simple modification techniques}
\label{appendix:modification}

For further improvement of Algorithm~\ref{alg}, we introduce two simple techniques in Algorithm~\ref{alg} to exploit adaptively the discrete information regarding $\mathfrak{A}^{N}$.

\noindent\textbf{(Iterative reweighting of cLiGME algorithm)}

The iterative reweighting technique (IW), e.g.,~\cite{candes_enhancing_2008}, has been used to enhance the effectiveness of the SOAV regularizer in~\eqref{eq:theta_SOAV} by updating the weights of the regularizer adaptively in an iterative algorithm.
IW is also used for Problem~\ref{prob:discrete}~\cite{hayakawa_convex_2017,hayakawa_reconstruction_2018}, where the SOAV regularizer in~\eqref{eq:theta_SOAV} is employed.
To utilize such a technique in Algorithm~\ref{alg}, we propose to set $\omega^{\ev{l}}_{n}$ $(l=1,2,\ldots,L;\,n=1,2,\ldots,N)$ in the seed functions $\norm{\cdot}_{\bm{\omega}^{\ev{l}};1}$ $(l=1,2,\ldots,L)$ adaptively by using the latest estimate ${\x}\coloneqq({x}_1,{x}_2,\ldots,{x}_{N})^{\top}$ as~\cite{hayakawa_reconstruction_2018}
\begin{equation}\label{eq:set_omega}
  \omega^{\ev{l}}_{n}=\frac{(\abs{{x}_{n}- a^{\ev{l}}}+\delta)^{-1}}{\sum_{l'=1}^{L} (\abs{{x}_{n}-a^{\ev{l'}}}+\delta)^{-1}},
\end{equation}
where $\delta\in\mathbb{R}_{++}$ is a sufficiently small number.
If $\abs{{x}_{n}- a^{\ev{l}}}$ is small, then
the corresponding $\omega^{\ev{l}}_{n}$ becomes large and ${x}_{n}$ will be close to $ a^{\ev{l}}$.
This IW can be realized by inserting
\begin{equation}\label{eq:iw}
  \begin{aligned}
     & \textbf{if~}k~\text{mod~}K=0\textbf{~then}                                                                                             \\
     & \hspace*{1em}\text{Update~}\bm{\omega}^{\ev{l}}=\left(\omega^{\ev{l}}_{1},\omega^{\ev{l}}_{2},\ldots,\omega^{\ev{l}}_{N}\right)^{\top} \\
     & \hspace*{1em}(l=1,2,\ldots,L)\text{~as~\eqref{eq:set_omega} with }\x=\x_k.                                                             \\
     & \textbf{end~if}
  \end{aligned}
\end{equation}
in line~4 of Algorithm~\ref{alg}, where $K\in\mathbb{N}$ controls the frequency of reweighting.

\noindent\textbf{(Generalized superiorization of cLiGME algorithm)}

\emph{Superiorization}~\cite{censor_perturbation_2010,fink_superiorized_2023} has been proposed as an idea to incorporate a favorable attribute into a given iterative algorithm, e.g., Krasnosel'ski\u{\i}-Mann iteration of Picard-type (see Fact~\ref{fct:picard}), without changing the inherent desired properties of the algorithm, by introducing perturbations into the iteration of the algorithm%
\footnote{
(Bounded perturbation resilience of Krasnosel'ski\u{\i}-Mann iteration~\cite{censor_perturbation_2010}).
Let $\mathcal{H}$ be a finite dimensional real Hilbert space, and $T:\mathcal{H}\to\mathcal{H}$ be an averaged nonexpansive operator.
Suppose $(\beta_{k})_{k=0}^{\infty}$ is a summable sequence in $\mathbb{R}_{+}$, and $(d_{k})_{k=0}^{\infty}$ is a bounded sequence in $\mathcal{H}$, where such a $(\beta_{k}d_{k})_{k=0}^{\infty}$ is said to be a sequence of bounded perturbations.
Then, with any initial point $u_{0}\subset\mathcal{H}$, $(u_{k})_{k=0}^{\infty}$ generated by
$(\forall k\in\mathbb{N})$ $u_{k+1}=T(u_{k}+\beta_{k} d_{k})$
converges to a point $\widebar{u}\in\Fix(T)$.
}.

We propose to incorporate a superiorization technique into Algorithm~\ref{alg} in order to move the estimate closer to $\mathfrak{A}^{N}$ at each iteration.
More precisely, we use a modification
\begin{equation}\label{eq:perturbation}
  \x_k\gets\x_k+\beta_k\underbrace{(P_{\mathfrak{A}^{N}}-\Id)(\x_{k})}_{\eqqcolon \d_{k}}
\end{equation}
in line~4 of Algorithm~\ref{alg}, where $(\beta_k)_{k=0}^{\infty}\subset\mathbb{R}_{+}$, and $(\d_{k})_{k=0}^{\infty}\subset\mathbb{R}^{N}$ is inspired by~\cite{fink_superiorized_2023} (see~\eqref{eq:projection} for $P_{\mathfrak{A}^{N}}$).
The global convergence to a fixed-point is guaranteed even by the  modification~\eqref{eq:perturbation} if $(\beta_{k})_{k=0}^{\infty}$ is summable and $(\d_{k})_{k=0}^{\infty}$ is bounded.

In this paper, we dare to propose to use more general $(\beta_{k})_{k=0}^{\infty}\subset\mathbb{R}_{+}$ which is not necessarily summable.
We call such a modification \emph{generalized superiorization}.
As shown in numerical experiments (see Section~\ref{sec:numerical_experiment}), the proposed generalized superiorization is effective to guide the sequence $(\x_{k})_{k=0}^{\infty}$ to the discrete set $\mathfrak{A}^{N}$.

\profile[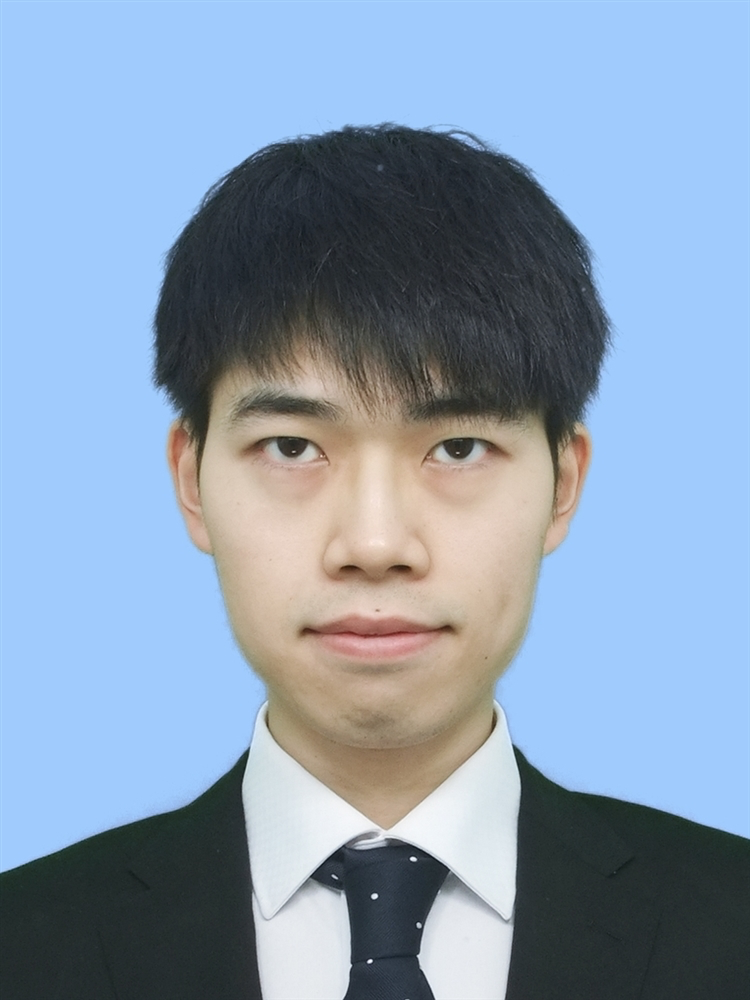]{Satoshi Shoji}{
  received the B.E. degree in information and communications engineering from Tokyo Institute of Technology in 2023.
  Currently, he has been working toward the M.E. degree with the Department of Information and Communications Engineering, Institute of Science Tokyo. 
  His current research interests are in signal processing and convex optimization.
}
\profile[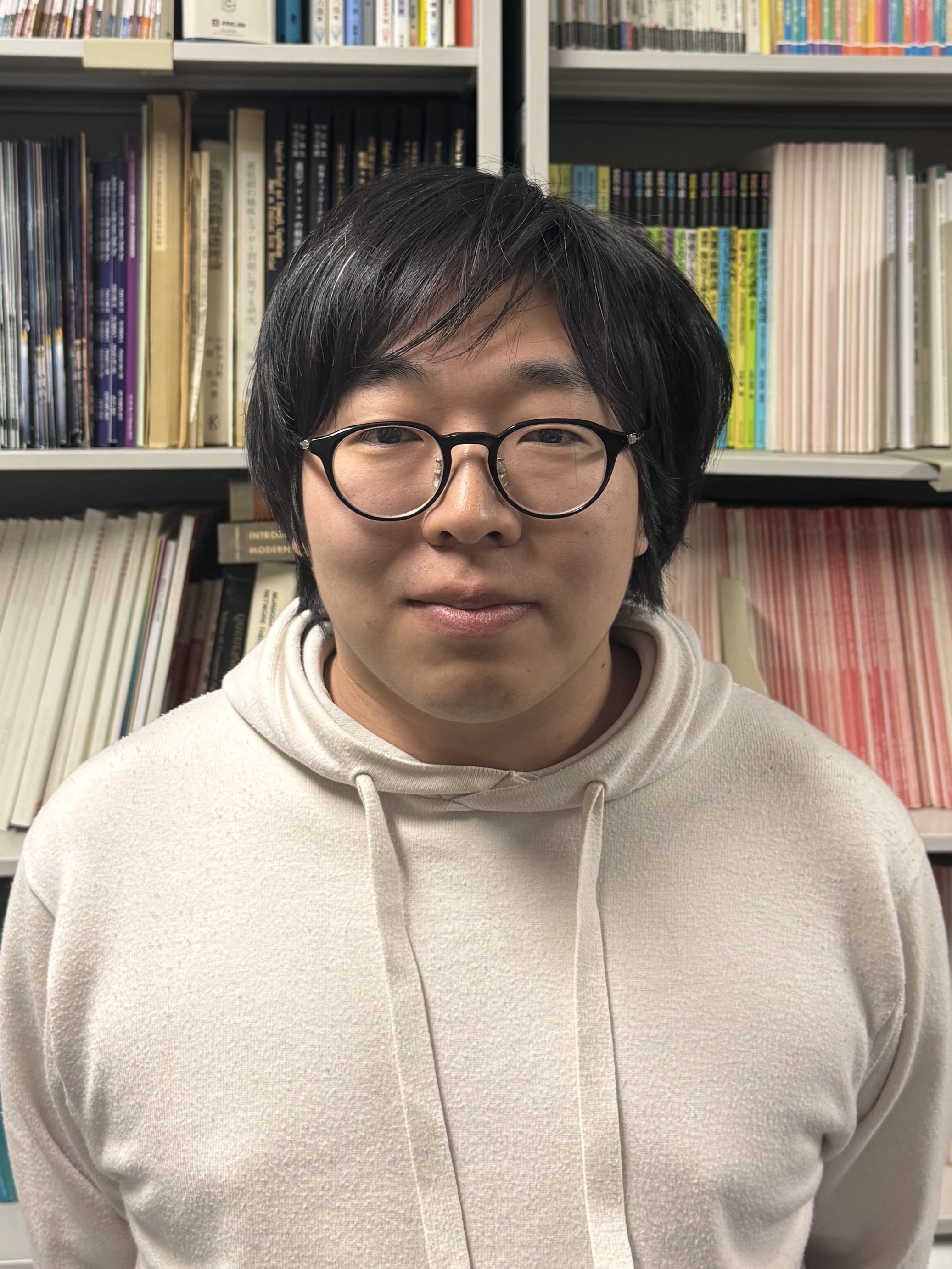]{Wataru Yata}{
  received the B.E. and M.E. degrees in information and communications engineering from the Tokyo Institute of Technology in 2021 and 2023 respectively. 
  Currently, he is a Ph.D. student in the Department of Information and Communications Engineering, Institute of Science Tokyo. 
  His current research interests are in signal processing and convex optimization.
}
\profile[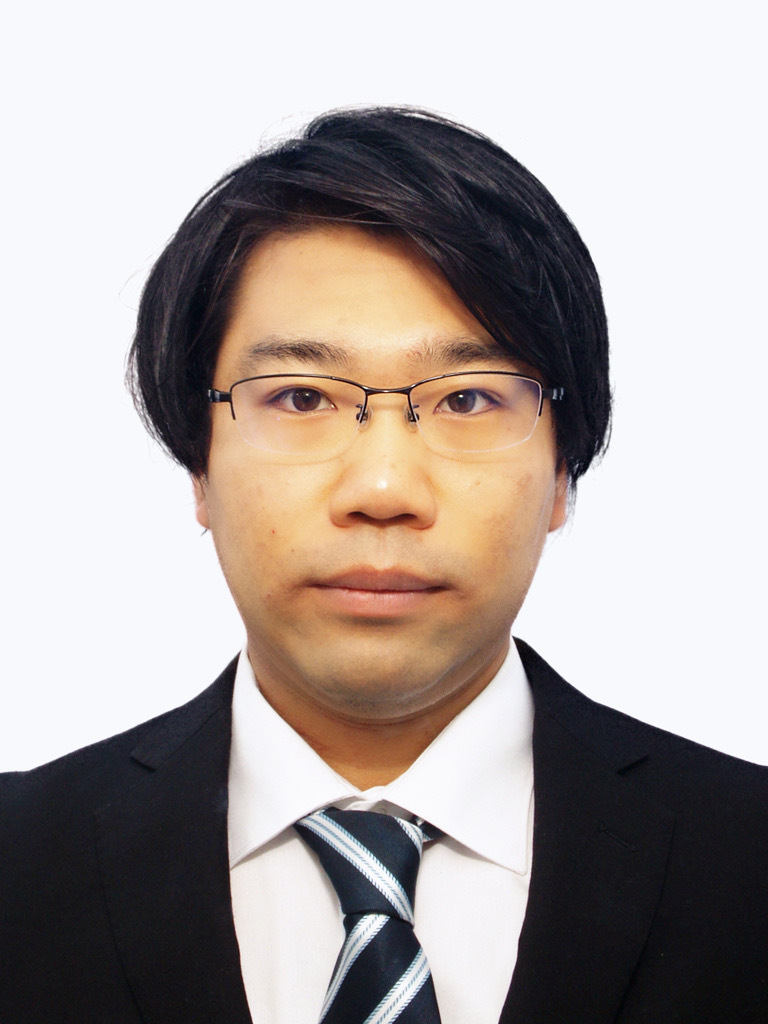]{Keita Kume}{
  received the B.E. degree in computer science and the M.E. and Ph.D. degrees in information and communications engineering from the Tokyo Institute of Technology in 2019, 2021, and 2024 respectively. 
  He is currently an assistant professor with the Dept. of Information and Communications Engineering, Institute of Science Tokyo. 
  He received Young Researcher Award from IEICE Technical Group on Signal Processing in 2021, and IEEE SPS Japan Student Conference Paper Award in 2024. 
}
\profile[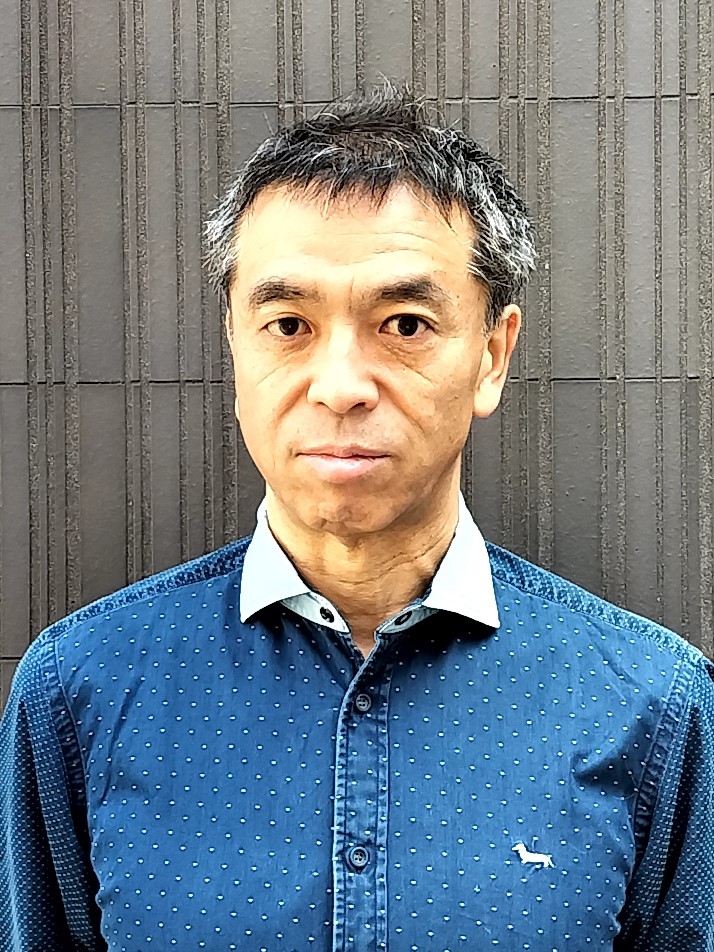]{Isao Yamada}{
  received the B.E. degree in computer science from the University of Tsukuba in 1985 and the M.E. and the Ph.D. 
  degrees in electrical and electronic engineering from the Tokyo Institute of Technology, in 1987 and 1990. 
  Currently, he is a professor with the Dept. of Information and Communications Engineering, Institute of Science Tokyo. 
  He has been the IEICE Fellow and IEEE Fellow since 2015.
}

\end{document}